\documentclass[twocolumn,10pt]{IEEEtran}
\usepackage{braket}

\input{def.tex}
\usepackage{dsfont}
\usepackage{minibox}
\DeclareSymbolFont{matha}{OML}{txmi}{m}{it}
\DeclareMathSymbol{\varv}{\mathord}{matha}{118}
\usepackage{eurosym}
\usepackage{hhline,physics}

\usepackage{multicol}
\usepackage{algorithm}
\usepackage{graphicx}
\usepackage{textcomp}
\usepackage{caption,pifont}
\usepackage[multiple]{footmisc}

\usepackage{algorithmic}

\newcommand{\trans}{^{\mathsf{T}}}
\newcommand{\herm}{^{\H}}
\makeatletter

\IEEEoverridecommandlockouts
\begin{document}
	\title{Cooperative RIS and STAR-RIS assisted mMIMO Communication: Analysis and Optimization} 
	\author{Anastasios Papazafeiropoulos, Ahmet M. Elbir, Pandelis Kourtessis, Ioannis Krikidis, Symeon Chatzinotas \thanks{Copyright (c) 2015 IEEE. Personal use of this material is permitted. However, permission to use this material for any other purposes must be obtained from the IEEE by sending a request to pubs-permissions@ieee.org.}
			\thanks{A. Papazafeiropoulos is with the Communications and Intelligent Systems Research Group, University of Hertfordshire, Hatfield AL10 9AB, U. K., and with SnT at the University of Luxembourg, Luxembourg. A. M. Elbir is with Interdisciplinary Centre for Security, Reliability, and Trust (SnT) at the University of Luxembourg, Luxembourg.  P. Kourtessis is with the Communications and Intelligent Systems Research Group, University of Hertfordshire, Hatfield AL10 9AB, U. K.  I. Krikidis is with the IRIDA Research Centre for Communication Technologies, Department of Electrical and Computer Engineering, University of Cyprus, Cyprus. S. Chatzinotas is with the SnT at the University of Luxembourg, Luxembourg. A. Papazafeiropoulos was supported  by the University of Hertfordshire's 5-year Vice Chancellor's Research Fellowship.  Also, this work was co-funded by the European Regional Development Fund and the Republic of Cyprus through the Research and Innovation Foundation under the project INFRASTRUCTURES/1216/0017 (IRIDA). It has also received funding from the European Research Council (ERC) under the European Union’s Horizon 2020 research and innovation programme (Grant agreement No. 819819).
			S. Chatzinotas   was supported by the National Research Fund, Luxembourg, under the project RISOTTI. E-mails: tapapazaf@gmail.com,ahmetmelbir@gmail.com, p.kourtessis@herts.ac.uk,krikidis.ioannis@ucy.ac.cy, symeon.chatzinotas@uni.lu.}}
	\maketitle\vspace{-1.7cm}
	\begin{abstract}
		Reconfigurable intelligent surface (RIS) has emerged as a cost-effective and promising solution to extend the wireless signal coverage and improve the performance via passive signal reflection. Different from existing works which do not account for the cooperation between RISs or do not provide full space coverage, we propose the marriage of cooperative double-RIS with simultaneously transmitting and reflecting RIS (STAR-RIS) technologies denoted as RIS/STAR-RIS under correlated Rayleigh fading conditions to assist the communication in a massive multiple-input multiple-output (mMIMO) setup. The proposed architecture is superior since it enjoys the benefits of the individual designs.  We introduce a channel estimation approach of the cascaded channels with reduced overhead. Also, we obtain the deterministic equivalent (DE) of the downlink achievable sum spectral efficiency (SE) in closed form based on large-scale statistics. Notably, relied on statistical channel state information (CSI), we optimise both surfaces by means of the projected gradient ascent method (PGAM), and obtain the gradients in closed form. The proposed optimization achieves to maximise the sum SE of such a complex system, and has low complexity and low overhead since it can be performed at every several coherence intervals. Numerical results show the benefit of the proposed architecture and  verify the analytical framework. In particular, we show that the RIS/STAR-RIS architecture outperforms the conventional double-RIS or its single-RIS counterparts.
	\end{abstract}
	
	\section{Introduction}
	Based on the recent  advancements in metasurfaces, reconfigurable intelligent surface  (RIS) has emerged as a promising technology to actualise a  smart and reconfigurable radio environment through passive   beamforming (PB) \cite{Wu2020,DiRenzo2020}. RIS consists of a large number of nearly passive elements, each of which can be independently adjusted to modify in real time the amplitude/phase shift of the impinging signal \cite{Wu2019}. The intelligent adjustment of its elements, enables RIS to proactively configure the wireless propagation channel towards better signal transmission instead of adapting to the channel by standard transceiver techniques. Among its benefits, we meet its low hardware cost as well as its lightweight and conformal geometry that can promote a large-scale and flexible deployment of RIS \cite{Wu2020,DiRenzo2020}.
	
	These attractive characteristics of RIS have attracted the interest from both academia and industry in studying the performance under various wireless system setups such as massive multiple-input-multiple-output (mMIMO) communication \cite{Papazafeiropoulos2022b,Papazafeiropoulos2022a}, orthogonal frequency 	division multiplexing (OFDM) \cite{Zheng2019,Yang2020b}, relaying communication \cite{Zheng2021a,Yildirim2021}, non-orthogonal multiple access (NOMA) \cite{Hou2020}, double/multi-RIS network \cite{Zheng2021b,Zheng2021,Mei2022,Dong2021,Abdullah2022,Papazafeiropoulos2022e,Ding2022,Han2022}. In particular, the double-RIS implementation is an interesting research direction providing enhanced coverage and better performance due to the multiplicated beamforming gain \cite{Zheng2021b}. However, most existing works on RIS have focused on setups with one or more  independently distributed RISs that serve user equipments (UEs) by  simple reflection only without considering the cooperation among multiple RISs. In the case of double-RIS, the PB over multiple RISs should be designed cooperatively, which means exploitation of the multiplicated beamforming gain while avoiding the undesired  interference to improve the system performance \cite{Zheng2021b,Zheng2021}. In \cite{Zheng2021b}, it was shown that a PB gain of order $ \mathcal{O}(N^{4}) $ can be obtained, where $ N $ is the total number of elements of the two cooperative RISs. This result outperforms the conventional single-RIS  with a PB gain  of order $ \mathcal{O}(N^{2}) $ \cite{Wu2019}. Notably, the implementation of two cooperative RISs induces additional path loss, which can be compensated by a sufficiently large number of $ N $. Although \cite{Han2020} presented this promising result,  it relied on an ideal line-of-sight (LoS) inter-RIS channel and a simplified system with a single antenna base station (BS), a single UE, and no single-reflection links. In \cite{Zheng2021b}, the problem of addressing a general setup with arbitrary channels and with multiple BS antennas/UEs was considered but does not provide full-space coverage,
	
	In parallel, most existing RIS research contributions assume that RISs can only reflect impinging waves which henceforth will be referred to as conventional single RISs. This design assumes that both the transmitter and the receiver are located on the same side of the surface, which results in half-space coverage and restricts the flexibility of the RIS technology since UEs may be located on both sides of the surface. To cover this gap, a new concept referred to as simultaneously transmitting and reflecting RIS (STAR-RIS) was presented in \cite{Xu2021,Mu2021,Papazafeiropoulos2022d,Papazafeiropoulos2022c,Niu2022,Xu2022}.  The works \cite{Zhang2022,Zeng2022} are among the first, where the feasibility of STAR-RIS was experimentally verified. According to STAR-RIS, the incident signal on a STAR-RIS element is divided into parts that correspond to the reflected and transmitted signals. In the former case, the signal is reflected to the same space as the impinging signal, while, in the latter case, the signal is transmitted to the opposite space. This is achieved by the manipulation of magnetic and electric currents of the surface elements that enable the reconfiguration of the transmitted and reflected signals through generally independent coefficients denoted as transmission and reflective coefficients  \cite{Xu2021,Mu2021}.  Despite its advantages, the literature concerning the integration of STAR-RIS into wireless  communication systems is still limited \cite{Xu2021,Mu2021,Papazafeiropoulos2022d,Papazafeiropoulos2022c,Niu2022,Xu2022}. For example, in \cite{Mu2021}, practical operation protocols for STAR-RIS have been proposed, and the joint transmission and reflection beamforming design  for unicast and multicast communication has been investigated. Also, in \cite{Papazafeiropoulos2022d}, the impact of  correlated Rayleigh fading was studied in STAR-RIS-assisted full duplex systems. Note that intelligent omni-surface (IOS), proposed in \cite{Zhang2020b}, is a similar idea to STAR-RIS, but the phase shifts for transmission and reflection  are identical. 
	
	\textit{Contributions}: The main contribution of this paper are summarised as follows:
	\begin{itemize}
		\item 	 Motivated by the above observations, we propose a double-RIS architecture, where the first and second surfaces consist of a conventional RIS and a STAR-RIS based on the energy splitting (ES) protocol, respectively. Contary to \cite{Zheng2021b}, the marriage of the double-RIS  with the STAR-RIS combines their advantages.  The double-RIS design aims mainly to extend the coverage, and the STAR-RIS aims  to provide $ 360^{\circ} $ coverage near the receiver side, while increasing the sum-rate. Specifically, two distributed surfaces  are deployed near the mMIMO BS and the group of nearby UEs to enhance communication. UEs can be located on both sides of the second surface (STAR-RIS). Note that we account for both single and double links as well as correlated Rayleigh fading.
		\item The introduction of the STAR-RIS in the double-RIS architecture emerges certain difficulties in the statistical analysis, channel estimation, and PB optimization. In particular, contrary to other works on STAR-RIS, we have  achieved  a unified analysis  regarding the channel estimation and data transmission phase that applies to a UE found in any of the $ t $ or $ r $ regions. Specifically, we consider an uplink training phase to obtain the imperfect CSI in closed-form for all links. Based on the deterministic equivalent (DE)
		analysis, which is often used in mMIMO systems, we derive in closed-form the downlink achievable signal-to-interference-plus-noise ratio (SINR) and the sum spectral efficiency (SE) being dependent on large-scale statistics. Also, contrary to many works that have relied on independent fading such as  \cite{Wu2019,Pan2020}, we have assumed  correlated fading, which is unavoidable in practice and affects the performance  \cite{Bjoernson2020}.	
		\item Given this system  and based on statistical CSI, we optimize the PBs at the two surfaces to maximise the same rate. To the best of our knowledge, we are the first to optimize simultaneously the amplitudes and the phase shifts of the PB in a STAR-RIS system.  This is a significant contribution since other works optimize only the phase shifts or optimize both the amplitudes and the phase shifts in an  alternating optimization manner. Notably, this property is  important for STAR-RIS applications, which have twice the number of optimization variables compared to reflecting-only RIS. In particular, since we can derive the gradients in closed form simple expressions, we apply the projected  gradient ascent method (PGAM)  alternatively for each surface, and obtain the corresponding optimal PB in closed form while achieving low complexity and low overhead. The main reason is that the optimization can take take place at every several coherence intervals when the statistical CSI changes.
		\item 	 We provide Monte-Carlo (MC) simulations to corroborate the theoretical results on the RIS/STAR-RIS   performance and investigate the effectiveness of the proposed optimization of PB. It is shown that the RIS/STAR-RIS can achieve significant performance gains   over the conventional double-RIS and single-RIS or STAR-RIS architectures.
	\end{itemize}

	\textit{Paper Outline}:	The rest of this paper is organised as follows. Section \ref{System} presents the system model of the proposed RIS/STAR-RIS architecture. In Section \ref{ChannelEstimation}, we provide the channel estimation approach. In Section \ref{DownlinkAchievable}, we derive the DE of the downlink sum SE.   In Section \ref{ProblemFormulation1}, we design the RIS PB   for each surface and study its performance. Simulations results are presented in Section \ref{numerical11} to assess the performance of the proposed design and validate the analytical. Finally, Section \ref{Conclusion} concludes the paper.
	
	\textit{Notation}: Vectors and matrices are denoted by boldface lower and upper case symbols, respectively. The notations $(\cdot)^\T$, $(\cdot)^\H$, and $\tr\!\left( {\cdot} \right)$ describe the transpose, Hermitian transpose, and trace operators, respectively. Moreover, the notations  $\EE\left[\cdot\right]$  and $ \mathrm{Var}(\cdot) $ express  the expectation and variance operators, respectively. The notation  $\diag\left(\bA\right) $ describes a vector with elements equal to the  diagonal elements of $ \bA $, the notation  $\diag\left(\bx\right) $ describes a diagonal  matrix whose elements are $ \bx $, while  $\bb \sim \cC\cN{(\b0,\mathbf{\Sigma})}$ describes a circularly symmetric complex Gaussian vector with zero mean and a  covariance matrix $\mathbf{\Sigma}$. 	Also, the notation  $a_n\asymp b_n$ with $a_n$ and $b_n$ being two infinite sequences denotes almost sure convergence as $ M \rightarrow \infty $, and the notation $  \pdv{f(x)}{x}   $ denotes the partial derivative of $ f $ with respect to  $ x $.

	\section{System Model}\label{System}
	We consider  a mixed double-RIS cooperatively assisted multi-user MIMO communication system, where two distributed RISs assist the communication from an $ M $-antenna mMIMO BS to $ K $ single-antenna UEs, as shown in Fig.~\ref{Fig1}.  The surface, positioned near the BS is denoted as RIS $ 1 $, and plays the role of an extender. The second surface, being a  STAR-RIS (denoted as RIS $ 2 $), can be located close to the UEs, which can be distributed on both sides of the surface such as indoor and outdoor UEs by providing $ 360^{\circ} $ coverage.\footnote{If the first reflection-only RIS is substituted by a STAR-RIS, the second STAR-RIS could be placed at the reflection or the transmission region, which leads to differences between these two scenarios. This research direction will be the topic of future research.}  In particular, $ \mathcal{K}_{r}\in \{1,\ldots,K_{r} \} $ UEs are  located in the reflection region $ (r) $ and $ \mathcal{K}_{t}\in \{1,\ldots,K_{t} \} $ UEs are  located in the transmission  region $ (t) $ of the STAR-RIS,  where $ K_{t}+K_{r}=K $. Moreover, we denote by $ \mathcal{W}_{k}=\{w_{1}, w_{2}, ..., w_{K}\} $  the  operation mode for the STAR-RIS  for each of the $ K $ UEs. Specifically, if the $ k $th UE is in the reflection region, i.e., $ k\in   \mathcal{K}_{r}$  then $ w_{k} = r $, otherwise, it will be  $ w_{k} = t $.  We assume that surfaces $ 1 $ and $ 2 $ consist of   uniform
	planar arrays (UPAs) of $ N_{1} $ and $ N_{2} $ passive elements, respectively, which means $ N_{1}+N_{2}=N $. The corresponding sets are denoted as $ \mathcal{N}_{1} $ and $ \mathcal{N}_{2} $.  Also, each distributed RIS is connected to a smart controller
	that adjusts its amplitudes/phase shifts and exchanges information with the BS via a separate reliable wireless link.  The RIS elements of both surfaces can be perfectly controlled. Despite the presence of blockages, we  assume the general scenario, where  direct links exist.

	\begin{figure}[!h]
		\begin{center}
			\includegraphics[width=0.9\linewidth]{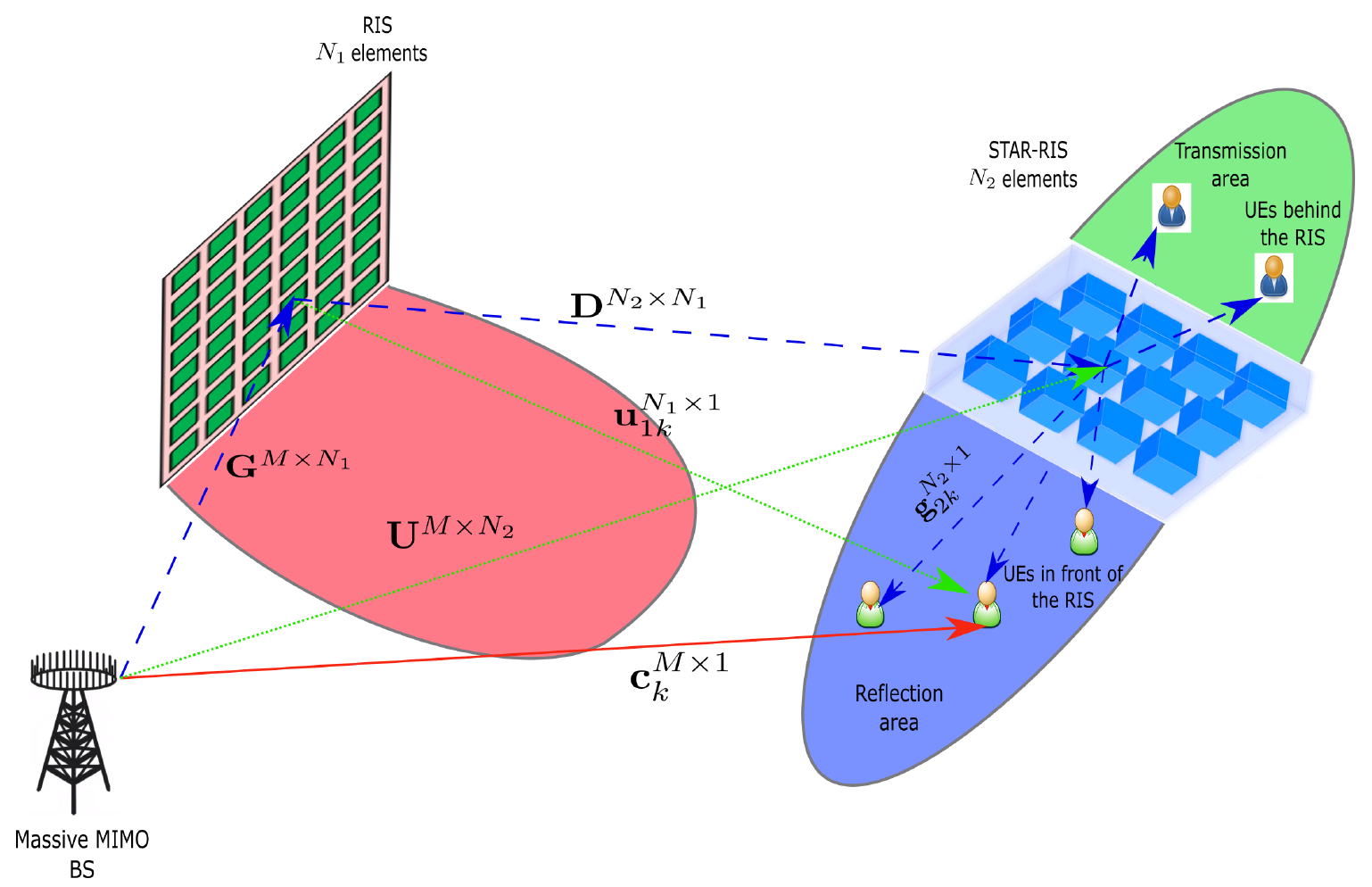}
			\caption{{ An mMIMO RIS/STAR-RIS assisted system with multiple UEs at transmission and reflection regions.  }}
			\label{Fig1}
		\end{center}
	\end{figure}

	We assume that the  STAR-RIS can configure  the transmitted ($ t $) and reflected ($ r $) signals by two independent coefficients. Specifically, let $ t_{n} =( {\beta_{n}^{t}}e^{j \phi_{n}^{t}})s_{n}$ and $ r_{n}=( {\beta_{n}^{r}}e^{j \phi_{n}^{r}})s_{n} $ denote the transmitted 	and reflected signal by the $ n $th  element of the STAR-RIS, respectively.  Regarding the amplitude and phase parameters, we have  $ {\beta_{n}^{w_{k}}}\in [0,1] $ and $ \phi_{i}^{w_{k}} \in [0,2\pi)$. According to this model  $ \phi_{n}^{t} $ and $ \phi_{n}^{r} $ can be chosen independently. However,  the choice of amplitudes is based on the relationship provided by the law of energy conservation as
	\begin{align}
		(\beta_{n}^{t})^{2}+(\beta_{n}^{r})^{2}=1,  \forall~ n =1,\ldots, {N}_{2}.
	\end{align}
	
	Furthermore, in the case of the conventional RIS, i.e., RIS $ 1 $, $\bPhi_{1}=\mathrm{diag}\left( \al_{1}e^{j \varphi_{1}}, \ldots, \al_{N}e^{j \varphi_{N}} \right)\in\mathbb{C}^{N_{1}\times N_{1}}$ is the diagonal PB that expresses the response of  $ N_{1} $ elements with $ \varphi_{n} \in [0,2\pi]$ and $ \al_{n}\in [0,1] $ describing the phase  and amplitude coefficient for element $ n=0,\ldots, N_{1} $, respectively. Without loss of generality, we make the common assumption of maximum  reflection ($ \al_{n}=1~ \forall n$) based on recent advances in  lossless metasurfaces \cite{Badloe2017}, because we want to focus on the operation and novel optimization of the STAR-RIS. The generalization to not optimal reflection for RIS $ 1 $ is straightforward and can follow similar lines to the STAR-RIS optimization.
	
	\subsection{Operation  for STAR-RIS}
	Our study concerns the ES protocol \cite{Mu2021}, which is summarized below. The study of the mode switching (MS) protocol is left for future work.
	
	\textit{ES protocol:} All  elements of STAR-RIS serve simultaneously  $ K $ UEs. Especially, the  PB is  expressed as $ \bPhi_{2,w_{k}}^{\mathrm{ES}}=\diag( {\beta_{1}^{w_{k}}}e^{j\phi_{1}^{w_{k}}}, \ldots,  {\beta_{N_{2}}^{w_{k}}}e^{j\phi_{N_{2}}^{w_{k}}}) \in \mathbb{C}^{N_{2}\times N_{2}}$, where $ \beta_{n}^{w_{k}} \ge 0 $, $ 		(\beta_{n}^{t})^{2}+(\beta_{n}^{r1})^{2}=1 $, and $ |e^{j\phi_{n}^{w_{k}}}|=1, \forall~ n =1,\ldots, {N}_{2} $.
	
	
	Henceforth, for the sake of exposition, we set $ \bar{\theta}_{n}=e^{j\varphi_{n}},  n\in \mathcal{ N}_{1} $ and $ \theta_{n}^{w_{k}}=e^{j\phi_{n}^{w_{k}}},  n\in \mathcal{ N}_{2} $. Also, we denote $\bar{\thetv}=\diag(\bPhi_{1})=[\bar{\theta}_{1}, \ldots, \bar{\theta}_{N_{1}}]^{\T}\in \mathbb{C}^{N_{1}\times 1} $, $\thetv^{u}=[\theta_{1}^{u}, \ldots, \theta^{u}_{N_{2}}]^{\T}\in \mathbb{C}^{N_{2}\times 1} $, and $\betv^{u}=[\beta^{u}_{1}, \ldots, \beta^{u}_{N_{2}}]^{\T}\in \mathbb{C}^{N_{2}\times 1} $, where $ u=\{t,r\} $.

	\subsection{Channel Model}\label{ChannelModel} 
	Based on a narrowband quasi-static block-fading model with independent channel realizations across different coherence blocks, we denote $ 	 \bG_{1}=[\bg_{11}\ldots,\bg_{1N_{1}} ] \in \mathbb{C}^{M \times N_{1}}$,   $ 	\bD\in \mathbb{C}^{N_{2} \times N_{1}}$, $ 	\bg_{2k}\in \mathbb{C}^{N_{2} \times 1}$, $ 	\bu_{1k}\in \mathbb{C}^{N_{1} \times 1}$,  $  \bU_{2}=[\bu_{21}\ldots,\bu_{2N_{2}} ] \in \mathbb{C}^{M \times N_{2}}$, and $ \bc_{k}\in \mathbb{C}^{M \times 1} $   as the channels  from the BS to  RIS, from  RIS to STAR-RIS, from  RIS  to UE $ k $,  from the BS  to the STAR-RIS, and the direct link  between the BS and UE  $ k $, respectively. By taking into account  correlated Rayleigh fading and path-loss,\footnote{The analysis, corresponding to correlated Rician fading, where an LoS component exists additionally to a multipath part,  is the topic of future work.} we have
	\begin{align}
		\mathrm{vec}(\bG_{1})&\sim \mathcal{CN}\left(\b0, \bR_{t1}\right),\label{eq1}\\
		\mathrm{vec}(\bD)&\sim \mathcal{CN}\left(\b0,\bR_{12} \right),\\
		\bg_{2k}&\sim \mathcal{CN}\left(\b0, \bR_{2k}\right),\\
		\bu_{1k}&\sim \mathcal{CN}\left(\b0, \bR_{1k}\right),\\
		\mathrm{vec}(\bU_{2})&\sim \mathcal{CN}\left(\b0, \bR_{t2}\right),\\
		\bc_{k}&\sim \mathcal{CN}\left(\b0, \bar{ \beta}_{k}\bR_{t}\right),\label{eq5}
	\end{align}
	where 
	$\bR_{t1}=\beta_{t1} \bR_{t}\otimes \bR_{1}\in \mathbb{C}^{MN_{1} \times M N_{1}} $,  $ \bR_{12}=\beta_{12}\bR_{1}\otimes \bR_{2}\in \mathbb{C}^{N_{1}N_{2} \times N_{1}N_{2}}  $, $\bR_{2k}=\beta_{2k} \bR_{2}\in \mathbb{C}^{N_{2} \times N_{2}} $, $\bR_{1k}=\beta_{1k} \bR_{1}\in \mathbb{C}^{N_{1} \times N_{1}} $, and $\bR_{t2}=\beta_{t2}  \bR_{t}\otimes \bR_{2}\in \mathbb{C}^{MN_{2} \times M N_{2}}  $ are the spatial covariance matrices of the respective links  with  $  \beta_{t1} $, $\beta_{12} $, $ \beta_{2k} $, $ \beta_{1k}$,  $  \beta_{t2}  $, and $ \bar{ \beta}_{k} $ being  the corresponding path-losses. Also, $ \bR_{t} \in \mathbb{C}^{M \times M}$ is the correlation matrix at the BS, while $ \bR_{1} \in \mathbb{C}^{N_{1} \times N_{1}}$ and $ \bR_{2} \in \mathbb{C}^{N_{2} \times N_{2}}$ are the correlation matrices of RIS  and  STAR-RIS, respectively.  Regarding $  \bR_{t} $, it can be modeled e.g., as in \cite{Hoydis2013}, and $ \bR_{1} $, $ \bR_{2} $ are modeled as in \cite{Bjoernson2020}, which describes isotropic Rayleigh fading. Specifically, let  $d_{q,\mathrm{V}}$ and $d_{q,\mathrm{H}}$ denote the vertical height and  horizontal width of each  element of RIS $ q=1,2 $, where RIS $ 1 $ and RIS $ 2 $ correspond to the conventional RIS and STAR-RIS, respectively. Then,  the $ \left(i,j\right) $th element of the representative correlation matrix $ \bR_{q} \in \mathbb{C}^{N_{q} \times N_{q}}  $ in \eqref{eq1}-\eqref{eq5} in the case of a RIS with $ N_{q}=N_{q,\mathrm{H}}N_{q,\mathrm{V}}  $ elements is given by
	\begin{align} \label{eq:Element}
		r_{q, ij} = d_{q,\mathrm{H}} d_{q,\mathrm{V}} \mathrm{sinc} \left( 2 \|\mathbf{u}_{i} - \mathbf{u}_{j} \|/\lambda\right),
	\end{align}
	where $\mathbf{u}_{\epsilon} = [0, \mod(\epsilon-1, N_{q,\mathrm{H}})d_{q,\mathrm{H}}, \lfloor (\epsilon-1)/N_{q,\mathrm{V}} \rfloor d_{q,\mathrm{V}}]^\T$, $\epsilon \in \{i,j\}$, and $\lambda$ is the wavelength of the plane wave, while    $ N_{q,\mathrm{H}} $ and $ N_{q,\mathrm{V}} $ denote the horizontally and  vertically  passive elements of RIS $ q $, i.e., $ N_{q}=N_{q,\mathrm{H}}\times N_{q,\mathrm{V}}$. Note that the path-losses and the covariance matrices are  assumed known  since they can be obtained with practical methods, {e.g., see} \cite{Neumann2018}.

	With  fixed PBs,  the aggregated channel vector for UE $ k $ via both RISs is 
	\begin{align}
		 \bh_{k}= \bd_{k}+ \bG_{1}\bPhi_{1}\bD\bPhi_{2,w_{k}} \bg_{2k},
	\end{align} 
and has a variance  $ \bR_{0k}=\EE\{\bh_{k}\bh_{k}^{\H}\} $. We have
	\begin{align}
		&\bR_{0k}=\bar{ \beta}_{k}\bR_{t}+\beta_{2k}\EE\{\bG_{1}\bPhi_{1}\bD\bPhi_{2,w_{k}}\bR_{2k}\bPhi_{2,w_{k}}^{\H}\bD^{\H}\bPhi_{1}^{\H}\bG_{1}^{\H}\}\label{cov1}\\
		&=\bar{ \beta}_{k}\bR_{t}+	\beta_{2k}\beta_{12}\tr(\bR_{1}\bPhi_{2,w_{k}}\bR_{2}\bPhi_{2,w_{k}}^{\H})\EE\{\bG_{1}\bPhi_{1}\bR_{2}\bPhi_{1}^{\H}\bG_{1}^{\H}\}\label{cov2}\\
		&=\bar{ \beta}_{k}\bR_{t}+	\hat{\beta}_{k}\tr(\bR_{1}\bPhi_{2,w_{k}}\bR_{2}\bPhi_{2,w_{k}}^{\H})\tr(\bR_{1}\bPhi_{1}\bR_{2}\bPhi_{1}^{\H})\bR_{t},\label{cov3}
	\end{align}
	where, in \eqref{cov1} we have relied on  the independence between $ \bc_{k} $, $ \bG_{1} $, $ \bD $, and $ \bg_{2k} $, and have applied $ \EE\{	\bg_{2k}	\bg_{2k}^{\H}\} = \beta_{2k} \bR_{\mathrm{RIS}}$. In \eqref{cov2}, we have expressed $ \bD $ in terms of its equivalent notation $ \bD=\sqrt{ \beta_{12}}\bR_{2}^{1/2}\tilde{\bD}\bR_{1}^{1/2} $ with $ \mathrm{vec}(\tilde{\bD})\sim \mathcal{CN}\left(\b0,\Id_{N_{1}N_{2}}\right) $, and  we have applied  the property $ \EE\{\bV \bU\bV^{\H}\} =\tr (\bU) \Id_{M}$ with $\bU  $ being a deterministic square matrix, and $ \bV $ being any matrix with independent and identically distributed (i.i.d.) entries of zero mean and unit variance. In \eqref{cov3}, we  have set $ 	\hat{\beta}_{k}=\beta_{t1}\beta_{2k}\beta_{12} $,  we have expressed $ \bG_{1} $ in terms of its equivalent notation $ \bG_{1}=\sqrt{ \beta_{t1}}\bR_{1}^{1/2}\tilde{\bG}\bR_{t}^{1/2} $ with $ \mathrm{vec}(\tilde{\bG})\sim \mathcal{CN}\left(\b0,\Id_{MN_{1}}\right) $, and we have applied again the previous property. It is worthwhile to mention that , when $\bR_{1} =\Id_{N_{1}} $, and  $\bR_{2}=\Id_{N_{2}} $,  $ \bR_{0k} $ does not depend on the phase shifts but only on the amplitudes, as also observed in \cite{Papazafeiropoulos2022}. In other words, $ \bR_{0k} $ can be optimized only with respect to the amplitudes.
	
	Given a fixed PB each time for the single links BS-STAR-RIS-UE $ k $ and BS-RIS-UE $ k $, the cascaded channels are $ \bh_{2k}=\bU_{2}\bPhi_{2,w_{k}}\bg_{2k} $ and $ \bh_{1k}=\bG_{1}\bPhi_{1}\bu_{1k} $ with variances  $ \bR_{2k}=\EE\{\bh_{2k}\bh_{2k}^{\H}\} $ and  $ \bR_{1k}=\EE\{\bh_{1k}\bh_{1k}^{\H}\} $, respectively. In particular, these can be written as
	\begin{align}
		\bR_{2k}&=\hat{\beta}_{2k}\tr(\bR_{t} \bPhi_{2,w_{k}}\bR_{2}\bPhi_{2,w_{k}}^{\H})\bR_{2},\label{cor2}\\
		\bR_{1k}&=\hat{\beta}_{1k}\tr(\bR_{t} \bPhi_{1}\bR_{1}\bPhi_{1}^{\H})\bR_{1},\label{cor1}
	\end{align}
	where $ \hat{\beta}_{2k}=\beta_{2k}\beta_{t2} $ and $ \hat{\beta}_{1k}=\beta_{1k}\beta_{t1} $. Note that we have used similar steps to \eqref{cov3} to obtain \eqref{cor2} and \eqref{cor1}, and have considered   $ \bU_{2} $ in terms of its equivalent notation $ \bU_{2}=\sqrt{ \beta_{t2}}\bR_{2}^{1/2}\tilde{\bU}\bR_{t}^{1/2} $ with $ \mathrm{vec}(\tilde{\bU})\sim \mathcal{CN}\left(\b0,\Id_{MN_{2}}\right) $.
	\begin{remark}\label{rem1}
		Under independent Rayleigh fading conditions, i.e., $ \bR_{1}=\Id_{N_{1}} $, $ \bR_{2}=\Id_{N_{2}} $, and $ \bR_{t}=\Id_{M} $, the variances of cascaded channels become $ \bR_{0k}=
		\hat{\beta}_{k}N_{1}\sum_{i=1}^{N_{2}}(\beta_{i}^{w_{k}})^{2}\Id_{M} $, $ \bR_{1k}=
		\hat{\beta}_{2k}\sum_{i=1}^{N_{2}}(\beta_{i}^{w_{k}})^{2}\Id_{N_{2}} $, and $ \bR_{2k}=
		\hat{\beta}_{2k}N_{1}\Id_{N_{1}} $, which are independent of the phase shifts. Hence, the correlation matrices appear dependence only on the amplitudes of the STAR-RIS. In this case, the optimization provided below, which is based on statistical CSI, takes place only with respect to the amplitudes of the second surface, while no phase shifts optimization can be performed. However, given that correlated fading is unavoidable in practice, the optimization of the  surfaces  depends on  the phase shifts in real-world scenarios.
	\end{remark}
	
	\section{Channel Estimation}\label{ChannelEstimation}
	In practice, CSI is imperfect. Herein, we rely on the  time division duplex (TDD) protocol to estimate the channels  by an uplink training phase with pilot symbols \cite{Bjoernson2017}. However, both RISs, consisting of nearly passive elements without any RF chains, cannot  obtain the received pilots by UEs and  process the estimated channels. Generally, there are two  approaches for channel estimation that correspond to  the estimation of the individual channels[3], [6], [11], and to the estimation of the aggregated channels by using long-term performance metrics  [4], [5], [36].  The advantages of the latter approach are  the implementation  without any  extra hardware and with low power cost.		 Hence, we follow the standard approach in mMIMO systems to estimate the cascaded channels, which requires less hardware and lower complexity compared to estimating the individual channels. Actually, one of the advantages of this work is to make the expression for the channel estimation "looking" identical for both types of users belonging to different areas of the second RIS but note that at the end it is different since $  {\bR}_{k} $ including the expressions corresponding to the   phases shifts is different, for users in  $ t $ and  $ r $ regions. In other words, we have achieved  to introduce the standard channel estimation for  multiple-user SIMO to STAR-RIS, which has not taken place before. 	However, there is a cost compared to the instantaneous performance. If we focused on  short-term metrics such as instantaneous SE, then the effective cascaded channels should be estimated by tuning the RIS passive beamforming (rather than fixed).\footnote{ In the case we would like to acquire the  individual channels, we could assume receive RF chains integrated into the RISs similar to \cite{Zheng2021,Zheng2021a}.} 
	
	Although the transmitted signal from  UE $ k $ can propagate along both channels simultaneously,  the double-RIS channel and single-RIS channel can be estimated separately. Specifically, first, the single-RIS channels are estimated by switching off the other RIS. Next, we can remove the single-RIS links by carefully 	selecting the pilots, as was performed  in \cite[Eq. 8]{Jiang2022}. Note that by picking good pilots, we can  also remove the double-RIS-aided link from the total channel without turning 		off the individual RIS as well.
	
	The following analysis achieves to obtain the estimated channels  for  fixed PBs in closed-forms in terms of large-scale statistics, and thus, channel estimation can be performed at every several coherence intervals.  Other methods such as \cite{He2019} do not provide analytical expressions while not capturing the correlation effect since they obtain the estimated channel per RIS element \cite{Nadeem2020}.

	In this direction, we assume that  each block has a duration of $\tau_{\mathrm{c}}$ channel uses, where $\tau$ channel uses are allocated for the uplink training phase and $\tau_{\mathrm{c}}-\tau$ channel uses are allocated for the downlink data transmission phase.

	For the double-RIS-assisted channels, we assume that all UEs either in  $ t $ or $ r $ region of the STAR-RIS send orthogonal pilot sequences. Given that the duration of the uplink training phase is $ \tau $ channel uses, we denote by $\bx_{k}=[x_{k,1}, \ldots, x_{k,\tau}]^{\T}\in \mathbb{C}^{\tau\times 1} $ the pilot sequence of UE $ k  $ that can be located in any of the two regions. Thus, we assume that all UEs from  both regions send pilots to the BS, which receives
	\begin{align}
		\bY^{\tr}=\sum_{i=1}^{K}\bh_{i}\bx_{i}^{\H} +
		\bZ^{\tr},\label{train1}
	\end{align}
	where  $ \bZ^{\tr} \in \mathbb{C}^{M \times \tau} $ is the received AWGN matrix having independent columns with each one distributed as $ \mathcal{CN}\left(\b0,\sigma^2\Id_{M}\right)$. Next, we multiply \eqref{train1} with the transmit training sequence from UE $ k $  to remove the  interference  from other UEs, which can be located in the same  or the opposite region, and we obtain
	\begin{align}
		\br_{k}=\bh_{k}+\frac{\bz_{k}}{ \tau P},\label{train2}
	\end{align}
	where  $ \bz_{k}=\bZ^{\tr} \bx_{k}$.  
	
	\begin{lemma}\label{PropositionDirectChannel}
		The linear  minimum mean square error (LMMSE) estimate of the  double-RIS-assisted channel $ \bh_{k} $ between  UE $ k $  and the  BS  is written as
		\begin{align}
			\hat{\bh}_{k}=\bR_{0k}\bQ_{0k} \br_{k},\label{estim1}
		\end{align}
		where $ \bQ_{0k}\!=\! \left(\!\bR_{0k}\!+\!\frac{\sigma^2}{ \tau P }\Id_{M}\!\right)^{\!-1}$, and $ \br_{k}$ is the noisy channel given by \eqref{train2}.
	\end{lemma}
	\begin{proof}
		Please see Appendix~\ref{lem1}.	
	\end{proof}
	
	According to the property of the orthogonality of LMMSE estimation, the overall perfect channel can be written  in terms of the estimated channel $\hat{\bh}_{k}$ and estimation channel error vectors $\tilde{\bh}_{k}$ as
	\begin{align}
		\bh_{k}=\hat{\bh}_{k}+\tilde{\bh}_{k}\label{current}, \end{align}
	where  $\hat{\bh}_{k}$ and $\tilde{\bh}_{k} $ are uncorrelated, have zero mean, and  variances  (cf. \eqref{var1}) 	$ \bPsi_{k}\!=\!\bR_{0k}\bQ_{0k}\bR_{0k} $ and $ 
	\bE_{k}=\bR_{0k}-\bPsi_{k} $, respectively.  
	
	Regarding the estimation of the single-RIS-assisted	channels, we provide the following lemma.
	\begin{lemma}\label{PropositionDirectChannel1}
		The LMMSE estimate of the  single-RIS-assisted channels $\bh_{ik} $ for $ i=1,2 $ between  UE $ k $  obeys to 
		\begin{align}
			\bh_{ik}=\hat{\bh}_{ik}+\tilde{\bh}_{ik}\label{current1}, ~i=1,2
		\end{align}
		where  $\hat{\bh}_{ik}$ and $\tilde{\bh}_{ik} $ are uncorrelated, have zero mean, and  variances  	$ \bPsi_{ik}\!=\!\bR_{ik}\bQ_{ik}\bR_{ik} $ and $ 
		\bE_{ik}=\bR_{ik}-\bPsi_{ik} $, respectively.
	\end{lemma}
	\begin{proof}
		The proof follows similar lines with the proof of Lemma~\ref{PropositionDirectChannel}.	
	\end{proof}

	\begin{remark}\label{rem3}
		We have followed a typical approach for channel estimation because it presents several advantages: i) This method  provides the estimated cascaded channel vectors in closed-forms, while other methods such as \cite{He2019} do not result in closed-form expressions; ii) Channel estimation can be performed at every several coherence intervals since the estimated channels depend  on large-scale statistics.
	\end{remark}

	With the above 	setup, the superimposed channel from the BS to  UE $ k $, which includes the double-reflection link (BS$ \rightarrow$ RIS $ \rightarrow$STAR-RIS$ \rightarrow $UE $ k $) and the two single-reflection links (BS$ \rightarrow $RIS$ \rightarrow $UE $ k $ ) and (BS$ \rightarrow $STAR-RIS$ \rightarrow $UE $ k $)
	is written as
	\begin{align}
		\bar{\bh}_{k}&=\bh_{k}+\bh_{1k}+\bh_{2k}\\
		&= \bc_{k}+ \bG_{1}\bPhi_{1}\bD\bPhi_{2,w_{k}} \bg_{2k}+\bG_{1}\bPhi_{1}\bu_{1k}+\bU_{2}\bPhi_{2,w_{k}}\bg_{2k}  \bh_{1k},\label{superimposed channel}
	\end{align}
	where $ \bar{\bh}_{k} $ has zero mean and variance $\bar{\bR}_{k}= \bR_{0k}+ \bR_{1k}+ \bR_{2k}$. Its LMMSE estimate is
	\begin{align}
		\hat{	\bar{\bh}}_{k}&=\hat{\bh}_{k}+\hat{\bh}_{1k}+\hat{\bh}_{2k},
	\end{align}
	which has zero mean and variance $ \bar{\bPsi}_{k}=\bPsi_{k}+\bPsi_{1k}+\bPsi_{2} $, while the estimation channel error vector $ \bar{\bee}_{k} $ has zero mean and variance $ \bE_{k}+ \bE_{1k}+ \bE_{2k}$. In other words, we have
	\begin{align}
		\bar{\bh}_{k}=\hat{	\bar{\bh}}_{k}+\bar{\bee}_{k}. \label{estimatedchannel}
	\end{align}

	\section{Downlink Achievable Rate}\label{DownlinkAchievable}
	During the downlink data transmission from the BS to UE $ k $ in $ t  $ or $ r  $ region, the received signal by UE $ k $ is expressed based on channel reciprocity as
	\begin{align}
		r_{k}=\bar{\bh}^\H_{k}\bs+z_{k},\label{DLreceivedSignal}
	\end{align}
	where   $\bs=\sqrt{\lambda} \sum_{i=1}^{K}\sqrt{p_{i}}\bff_{i}l_{i}$ denotes the transmit signal vector  by the BS. Herein, $ p_{i} $ is  the power allocated to UE $ i $,  and $ \lambda $ is a constant which is found such that $ \EE[\mathbf{s}^{\H}\mathbf{s}]=\rho $, where $ \rho $ is the total average power budget.\footnote{Herein, we rely on a common assumption in the mMIMO literature, which is the adoption of equal power allocation among all UEs, i.e., $ p_i = \rho/K $ \cite{Hoydis2013}.} Note that $ \lambda $ is given by $ \lambda=\frac{K}{\EE\{\tr(\bF\bF^{\H})\}} $ to guarantee $  \EE[\mathbf{s}^{\H}\mathbf{s}]=\rho$,
	where $ \bF=[\bff_{1}, \ldots, \bff_{K}] \in\mathbb{C}^{M \times K}$. Moreover, $z_{k} \sim \cC\cN(0,\sigma^{2})$ is the additive  white complex Gaussian noise at UE $k$. Also,   $\bff_{i} \in \bbC^{M \times 1}$ is the linear precoding vector and $ l_{i} $ is     the corresponding data symbol with $ \EE\{|l_{i}|^{2}\}=1 $.

	Taking advantage of the technique in~\cite{Medard2000} and by exploiting that UEs do not have instantaneous CSI but are aware of only statistical CSI, the downlink SINR can be written as\footnote{We would like to mention that (25) is a known lower bound in the mMIMO literature, which is quite accurate according to the relevant literature, e.g., please see   [32], [35]. Apart from this, we have provided Monte-Carlo manipulations in Figures 2-4  that verify both the correctness and accuracy of the analytical results including this bound.}
	\begin{align}
		\gamma_{k}=	\frac{	S_{k}}{I_{k}},\label{sinr1}
	\end{align}
	where 	
	\begin{align}
		S_{k}&=|\EE\{\bar{\bh}_{k}^{\H}\bff_{k}\}|^{2} \label{Sig} \\
		I_{k}&=\EE\big\{ \big| \bar{\bh}_{k}^{\H}{\bff}_{k}-\EE\big\{
		\bar{\bh}_{k}^{\H}{\bff}_{k}\big\}\big|^{2}\big\}\!+\!\sum_{i\ne k}^{K}|\EE\{\bar{\bh}^\H_{k}\bff_{i}\}|^{2}\!+\!\frac{K\sigma^{2}}{\rho \lambda}.\label{Int}\end{align}
	
	Although both maximum ratio transmission (MRT) and  regularized zero-forcing (RZF) precoders are common options in the mMIMO literature for the downlink transmission, we select MRT for the sake of simplicity while RZF will be investigated in future work, i.e., $ \mathbf{f}_k = \hat{\bar{\mathbf{h}}}_k $.
	
	The following analysis requires $ M $, $ N $, and $ K $ increase  but with a given bounded ratio as $ 0 < \lim \inf \frac{K
	}{M}\le \lim \sup \frac{K
	}{M}< \infty $  and $ 0 < \lim \inf \frac{N
	}{M}\le \lim \sup \frac{
	}{M}< \infty $. Henceforth, this notation is denoted as $ \asymp$. Also, the covariance matrices obey similar assumptions provided in \cite[Assump. A1-A3]{Hoydis2013}. Note that the DE analysis allows obtaining  deterministic expressions, which makes lengthy Monte-Carlo simulations unnecessary. In parallel, deterministic expressions are tight approximations even for conventional systems with moderate  dimensions, e.g., an $ 8 \times 8 $ matrix~\cite{Couillet2011}.

	\begin{proposition}\label{Proposition:DLSINR}
		The downlink achievable SINR of UE $k$ with MRT precoding for  given PBs $ \bPhi_{1} $ and $ \Phi_{2, w_k} $  in a RIS/STAR-RIS assisted mMIMO system, accounting for imperfect CSI and correlated Rayleigh fading, is given by \eqref{sinr1}, 		where
		\begin{align}
			{S}_{k}&\asymp\tr^{2}\left(\bar{\bPsi}_{k}\right)\!,\label{Num1}\\
			{I}_{k}&\asymp\sum_{i =1}^{K}\tr\!\left(\bar{\bR}_{k}\bar{\bPsi}_{i} \right)-\tr\left( \bar{\bPsi}_{k}^{2}\right)+\frac{K\sigma^{2}}{ \rho}\sum_{i=1}^{K}\tr(\bar{\bPsi}_{i}).\label{Den1}
		\end{align}
	\end{proposition} 
	\begin{proof}
		Please see Appendix~\ref{Proposition1}.	
	\end{proof}
	
	The downlink achievable  sum SE is given by
	\begin{align}
		\mathrm{SE}	=\frac{\tau_{\mathrm{c}}-\tau}{\tau_{\mathrm{c}}}\sum_{k=1}^{K}\log_{2}\left ( 1+\gamma_{k}\right)\!,\label{LowerBound}
	\end{align}
	where   the pre-log fraction describes   the percentage of samples per coherence block  for downlink data transmission.
	
	\section{Problem Formulation for RIS and STAR-RIS}\label{ProblemFormulation1}
	
	Based on the common assumption of infinite resolution phase shifters, we propose an alternating optimization algorithm for designing the cooperative reflecting beamforming by optimizing the sum SE with imperfect CSI and correlated fading, which is formulated as
	\begin{equation}
		\begin{IEEEeqnarraybox}[][c]{rl}
			\max_{\bar{\thetv},\thetv,\betv}&\quad\mathrm{SE}(\bar{\thetv},\thetv,\betv)\\
			\mathrm{s.t}&	\quad|\bar{\theta}_{n}|=1, ~\forall n \in \mathcal{N}_{1}\\
			&\quad (\beta_{n}^{t})^{2}+(\beta_{n}^{r})^{2}=1,  \forall n \in \mathcal{N}_{2}\\
			&\quad\beta_{n}^{t}\ge 0, \beta_{n}^{r}\ge 0,~\forall n \in \mathcal{N}_{2}\\
			&\quad|\theta_{n}^{t}|=|\theta_{n}^{r}|=1, ~\forall n \in \mathcal{N}_{2}
		\end{IEEEeqnarraybox}\label{Maximization}\tag{$\mathcal{P}1$}
	\end{equation}
	where $\thetv=[(\thetv^{t})^{\T}, (\thetv^{r})^{\T}]^{\T} $ and $\betv=[(\betv^{t})^{\T}, (\betv^{r})^{\T}]^{\T} $. We have put the variables to be optimized in parentheses to emphasize their presence.

	Obviously, the  problem \eqref{Maximization} is non-convex. Also, a coupling among the optimization variables appears, which are the phase-shifts of  RIS 1 as well as the amplitudes and the phase shifts of the STAR-RIS for transmission and reflection. For the sake of exposition, we define the following three sets describing the feasible set of \eqref{Maximization}: $ \bar{\Theta}=\{\bar{\thetv}\ |\ |\bar{\theta}_{i}=1,i=1,2,\ldots N_{1}\} $, 
	$ \Theta=\{\thetv\ |\ |\theta_{i}^{t}|=|\theta_{i}^{r}|=1,i=1,2,\ldots N_{2}\} $, and $ \mathcal{B}=\{\betv\ |\ (\beta_{i}^{t})^{2}+(\beta_{i}^{r})^{2}=1,\beta_{i}^{t}\geq0,\beta_{i}^{r}\geq0,i=1,2,\ldots N_{2}\} $. We observe that the project operators of the above sets can be obtained in closed-form, which motivates the application of the  PGAM \cite[Ch. 2]{Bertsekas1999}  for the optimization of $\bar{\thetv}$,  $\thetv$, and $\betv$ as follows. In particular, since the optimization problem includes two PBs, which are $ \bPhi_{1} $ and $ \bPhi_{2,w_{k}} $, we perform alternating optimization by keeping one PB fixed while optimizing the other  in an iterative manner until reaching  convergence to a stationary point.
	\begin{remark}
		Under independent Rayleigh fading conditions and according to Remark \ref{rem1}, the sum rate in \eqref{LowerBound} is independent of the phase shifts of the RIS and the STAR-RIS, but it can be optimized with respect to the amplitudes of the STAR-RIS.
	\end{remark}
	
	Thus, we have to apply PGAM two times, one for each surface. Note that each algorithm below achieves a local optimum, i.e., different initializations provide different solutions. In other words, the overall algorithm, based on alternating optimization, results in a local optimum. With this procedure and given the total average power budget constraint regarding \eqref{Den1}, the sum-rate increases  until convergence.

	The proposed algorithms below, i.e., Algorithms 1 and 2 converge quickly and have low computation complexity. Moreover, given that both algorithms achieve a local optimum and that the overall algorithm is based on alternating optimization (AO), the final solution corresponds to a local optimum, which means that different initializations will result in different solutions, as will be shown below in Sec. VI.
	
	\subsection{Problem Formulation for RIS}\label{ProblemFormulation2}
	The formulation problem for  RIS $ 1 $ is written as
	\begin{align}
		\begin{IEEEeqnarraybox}[][c]{rl}
			\max_{\bar{\thetv}}&\quad\mathrm{SE}(\bar{\thetv})\\
			\mathrm{s.t}&	\quad|\bar{\theta}_{n}|=1, ~\forall n \in \mathcal{N}_{1}\\
		\end{IEEEeqnarraybox}\label{Maximization1}\tag{$\mathcal{P}2$}, 
	\end{align}
	where we define the set $ \bar{\Theta}=\{\bar{\thetv}\ |\ |\bar{\theta}_{i}|=1,i=1,\ldots N_{1}\} $.
	
	Given the non-convexity of \eqref{Maximization1} in terms of $ \bPhi_{1}  $ and with a unit-modulus constraint regarding $ \bar{\theta}_{n} $,  PGAM is a very good candidate as mentioned. Hence, the first PGAM, concerning the RIS,  includes the following iteration
	\begin{align}
		\bar{\thetv}^{n+1}&=P_{\bar{\Theta}}(\bar{\thetv}^{n}+\mu_{1,n}\nabla_{\bar{\thetv}}\mathrm{SE}(\bar{\thetv}^{n})).\label{step1}  
	\end{align}
	
	The superscript expresses the iteration count while we move toward the gradient direction to increase the objective. Note that $\mu_{1,n}$ is the step size for  $\bar{\thetv}$. The difficulty of the problem does not allow to obtain the ideal step size, which should be equal to the inversely proportional of the Lipschitz constant. Thus, to find the step size at each iteration, we apply the Armijo-Goldstein backtracking line search by defining  a quadratic approximation of $\mathrm{SE}(\bar{\thetv})$ as
	\begin{align}
		&	\bar{Q}_{\mu_{1}}(\bar{\thetv} ;\bx)=\mathrm{SE}(\bar{\thetv})+\langle	\nabla_{\bar{\thetv}}\mathrm{SE}(\bar{\thetv}),\bx-\bar{\thetv}\rangle-\frac{1}{\mu_{1}}\|\bx-\bar{\thetv}\|^{2}_{2}.
	\end{align}    
	
	The step size $  \mu_{1,n} $ in \eqref{step1}, used at iteration $n$ as the initial step size at iteration $n+1$, can be obtained as $ \mu_{1,n} = L_{1,n}\kappa_{1}^{m_{1,n}} $, where $ L_{1,n}>0 $,  $ \kappa_{1} \in (0,1) $, and $ m_{1,n} $ is the
	smallest nonnegative integer satisfying
	\begin{align}
		\mathrm{SE}(\bar{\thetv}^{n+1})\geq	Q_{L_{n}\kappa^{m_{n}}}(\bar{\thetv}^{n};\bar{\thetv}^{n+1}),
	\end{align}
	which is performed by an iterative procedure. The proposed PGAM is described in Algorithm \ref{Algoa1}. 
	\begin{algorithm}[th]
		\caption{Projected Gradient Ascent Algorithm for  RIS 1 Design\label{Algoa1}}
		\begin{algorithmic}[1]
			\STATE Input: $\bar{\thetv}^{0},\mu_{1,n}>0$, $\kappa_{1}\in(0,1)$
			\STATE $n\gets1$
			\REPEAT
			\REPEAT \label{ls:start}
			\STATE $\bar{\thetv}^{n+1}=P_{\bar{\Theta}}(\bar{\thetv}^{n}+\mu_{1,n}\nabla_{\bar{\thetv}}\mathrm{SE}(\bar{\thetv}^{n}))$
			\IF{ $\mathrm{SE}(\bar{\thetv}^{n+1})\leq \bar{Q}_{\mu_{1,n}}(\bar{\thetv}^{n};\bar{\thetv}^{n+1})$}
			\STATE $\mu_{1,n}=\mu_{1,n}\kappa_{1}$
			\ENDIF
			\UNTIL{ $\mathrm{SE}(\bar{\thetv}^{n+1})>\bar{Q}_{\mu_{1,n}}(\bar{\thetv}^{n};\bar{\thetv}^{n+1})$}\label{ls:end}
			\STATE $\mu_{1,n+1}\leftarrow\mu_{1, n}$
			\STATE $n\leftarrow n+1$
			\UNTIL{ convergence}
			\STATE Output: $\bar{\thetv}^{n+1}$
		\end{algorithmic}
	\end{algorithm} 
	
	\begin{proposition}\label{PropGradients}
		The complex gradient $ \nabla_{\bar{\thetv}}\mathrm{SE}(\bar{\thetv}) $ is given in closed-form by
		\begin{align}
			\nabla_{\bar{\thetv}}\mathrm{SE}(\bar{\thetv})&=\frac{\tau_{\mathrm{c}}-\tau}{\tau_{\mathrm{c}}\log2}\sum_{k=1}^{K}\frac{	I_{k}\nabla_{\bar{\thetv}}{S_{k}}-S_{k}	\nabla_{\bar{\thetv}}{I_{k}}}{(1+\gamma_{k})I_{k}^{2}},
		\end{align}
		where
		\begin{align}
		&	\nabla_{\bar{\thetv}}S_{k}=\nu_{0k}\diag\bigl(\bB_{1}\bigr)+\nu_{1k}\diag\bigl(\bB_{2}\bigr) \label{derivbartheta},\\
		&	\nabla_{\bar{\thetv}}I_k=\frac{\partial}{\partial\boldsymbol{\bar{\theta}}^{\ast}}I_{k} \nn\\
			& =\diag\bigl(\bar{\nu}_{1k}\bB_{1}+\bar{\nu}_{2k}\bB_{2}+\sum\nolimits_{i=1}^{K}(\tilde{\nu}_{ki1}\bB_{1}+\tilde{\nu}_{ki2}\bB_{2})\bigr),
		\end{align}
		with $ \bB_{1}=\bR_{1}\bPhi_{1}\bR_{2} $, $ \bB_{2}=\bR_{t} \bPhi_{1}\bR_{1} $, $ \nu_{0k}=2\tr(\bar{\boldsymbol{\Psi}}_{k})\tr\bigl(\bar{C}\bQ_{0k} \bR_{0k} 	\bR_{t}-\bar{C}\bQ_{1k}\bR_{0k}^{2}\bQ_{1k}	\bR_{t}+\bar{C}\bR_{0k}\bQ_{0k}	\bR_{t}\bigr) $, $ \nu_{1k}=2\tr(\bar{\boldsymbol{\Psi}}_{k})\tr\bigl(\hat{\beta}_{1k}\bQ_{1k}\bR_{1k}\bR_{1}-\hat{\beta}_{1k}\bQ_{1k}\bR_{1k}^{2}\bQ_{1k}\bR_{1}+\hat{\beta}_{1k}\bR_{1k}\bQ_{1k}\bR_{1}\bigr) $, $\bar{\nu}_{1k}=\bar{C}	\tr\bigl(\check{\bar{\boldsymbol{\Psi}}}_{k}\bR_{t}\bigr)$,$\bar{\nu}_{2k}=\hat{\beta}_{1k}\tr\bigl(\check{\bar{\boldsymbol{\Psi}}}_{k}\bR_{1}\bigr)$, $\tilde{\nu}_{ki1}=\bar{C}\tr\bigl(\tilde{\mathbf{R}}_{ki}\bR_{t}\bigr)$, and $\tilde{\nu}_{ki2}=\hat{\beta}_{1k}\tr\bigl(\tilde{\mathbf{R}}_{ki}\bR_{1}\bigr)$.
	\end{proposition}
	\begin{proof}
		Please see Appendix~\ref{Prop2}.	
	\end{proof}
	
	Note that the projection onto the set  $ \bar{\Theta} $ for a given $\bar{\thetv}\in \mathbb{C}^{2N_{1}\times 1}$   is given by 
	\begin{equation}
		P_{\bar{\Theta}}(\bar{\thetv})=\bar{\thetv}/|\bar{\thetv}|=e^{j\angle\bar{\thetv}}.
	\end{equation}
	with the operations in the right-hand taking place entry-wise.

	\subsection{Problem Formulation for STAR-RIS}\label{ProblemFormulation3}
	In the interesting case of the STAR-RIS, the optimization problem reads as
	\begin{equation}
		\begin{IEEEeqnarraybox}[][c]{rl}
			\max_{\thetv,\betv}&\quad\mathrm{SE}(\thetv,\betv)\\
			&\quad (\beta_{n}^{t})^{2}+(\beta_{n}^{r})^{2}=1,  \forall n \in \mathcal{N}_{2}\\
			&\quad\beta_{n}^{t}\ge 0, \beta_{n}^{r}\ge 0,~\forall n \in \mathcal{N}_{2}\\
			&\quad|\theta_{n}^{t}|=|\theta_{n}^{r}|=1, ~\forall n \in \mathcal{N}_{2}
		\end{IEEEeqnarraybox}\label{Maximization2}\tag{$\mathcal{P}3$},
	\end{equation}
	where the   feasible set of \eqref{Maximization2} can be  defined for the sake of exposition by the sets  $ \Theta=\{\thetv\ |\ |\theta_{i}^{t}|=|\theta_{i}^{r}|=1,i=1,2,\ldots N\} $, and $ \mathcal{B}=\{\betv\ |\ (\beta_{i}^{t})^{2}+(\beta_{i}^{r})^{2}=1,\beta_{i}^{t}\geq0,\beta_{i}^{r}\geq0,i=1,2,\ldots N\} $. This problem is non-convex, and includes coupling among the amplitudes and the phase shifts for transmission and reflection. Thus, we suggest the application of PGAM, which  includes the following iterations
	\begin{subequations}\label{mainiteration}\begin{align}
			\thetv^{n+1}&=P_{\Theta}(\thetv^{n}+\mu_{2,n}\nabla_{\thetv}\mathrm{SE}(\thetv^{n},\betv^{n})),\label{step2} \\ \betv^{n+1}&=P_{\mathcal{B}}(\betv^{n}+{\mu}_{2,n}\nabla_{\betv}\mathrm{SE}(\thetv^{n},\betv^{n})),\label{step3} \end{align}
	\end{subequations}
	where the step size  $  \mu_{2,n} $  can be obtained as $ \mu_{2,n} = L_{2,n}\kappa_{2}^{m_{2,n}} $ with $ m_{2,n} $ being the smallest nonnegative integer satisfying
	\begin{align}
		\mathrm{SE}(\thetv^{n+1},\betv^{n+1})\geq	Q_{L_{2,n}\kappa_{2}^{m_{2,n}}}(\thetv^{n}, \betv^{n};\thetv^{n+1},\betv^{n+1}).
	\end{align}
	
	Herein, the selection of the step size in \eqref{step1} and \eqref{step2} is crucial for the convergence of the PGAM. Again, we resort to the application of the Armijo-Goldstein backtracking line search to acquire the step size. Hence, we define the quadratic approximation
	\begin{align}
		&	Q_{\mu_{2}}(\thetv, \betv;\bx,\by)=\mathrm{SE}(\thetv,\betv)+\langle	\nabla_{\thetv}\mathrm{SE}(\thetv,\betv),\bx-\thetv\rangle\nn\\
		&-\frac{1}{\mu_{2}}\|\bx-\thetv\|^{2}_{2}+\langle\nabla_{\betv}\mathrm{SE}(\thetv,\betv),\by-\betv\rangle-\frac{1}{\mu_{2}}\|\by-\betv\|^{2}_{2}.
	\end{align}

	The suggested  PGAM is outlined in Algorithm \ref{Algoa2}. 
	\begin{algorithm}[th]
		\caption{Projected Gradient Ascent Algorithm for the STAR-RIS Design\label{Algoa2}}
		\begin{algorithmic}[1]
			\STATE Input: $\thetv^{0},\betv^{0},\mu_{2,n}>0$, $\kappa_{2}\in(0,1)$
			\STATE $n\gets1$
			\REPEAT
			\REPEAT \label{ls:start}
			\STATE $\thetv^{n+1}=P_{\Theta}(\thetv^{n}+\mu_{2,n}\nabla_{\thetv}\mathrm{SE}(\thetv^{n},\betv^{n}))$
			\STATE $\betv^{n+1}=P_{B}(\betv^{n}+\mu_{2,n}\nabla_{\betv}\mathrm{SE}(\thetv^{n},\betv^{n}))$
			\IF{ $\mathrm{SE}(\thetv^{n+1},\betv^{n+1})\leq Q_{\mu_{2,n}}(\thetv^{n},\betv^{n};\thetv^{n+1},\betv^{n+1})$}
			\STATE $\mu_{2,n}=\mu_{2,n}\kappa_{2}$
			\ENDIF
			\UNTIL{ $\mathrm{SE}(\thetv^{n+1},\betv^{n+1})>Q_{\mu_{2,n}}(\thetv^{n},\betv^{n};\thetv^{n+1},\betv^{n+1})$}\label{ls:end}
			\STATE $\mu_{2,n+1}\leftarrow\mu_{2,n}$
			\STATE $n\leftarrow n+1$
			\UNTIL{ convergence}
			\STATE Output: $\thetv^{n+1},\betv^{n+1}$
		\end{algorithmic}
	\end{algorithm} 
	
	The description of Algorithm \ref{Algoa2} concludes by providing the projection onto the sets $ \Theta $ and $ \mathcal{B} $. The former is given for a given $ \Theta $ and $ \mathcal{B} $ by the entry-wise operation as
	\begin{align}
		P_{\Theta}(\thetv)=\thetv/|\thetv|=e^{j\angle\thetv}.
	\end{align}
	
	In the case of $P_{ \mathcal{B} }(\betv)$, we observe that the  constraint $(\beta_{i}^{t})^{2}+(\beta_{i}^{r})^{2}=1,\beta_{i}^{t}\geq0,\beta_{i}^{r}\geq0$ corresponds to  the first quadrant of the  unit circle, which makes $P_{ \mathcal{B} }(\betv)$ complicated. To improve the efficiency of $P_{ \mathcal{B} }(\betv)$ during the iterative process, we allow   $\beta_{i}^{t}$ and $\beta_{i}^{r}$ to take negative value without affecting the optimality of the proposed solution, while achieving the same objective. Hence,  after projecting $\beta_{i}^{t}$ and $\beta_{i}^{r}$ onto the entire unit circle,  $P_{ \mathcal{B} }(\betv)$ is written as 
	\begin{subequations}
		\begin{align}
			\left[\ensuremath{P_{\mathcal{B}}(}\boldsymbol{\beta})\right]_{i} & =\frac{{\beta}_{i}}{\sqrt{{\beta}_{i}^{2}+{\beta}_{i+N}^{2}}},i=1,2,\ldots,N_{2}\\
			\left[\ensuremath{P_{\mathcal{B}}(}\boldsymbol{\beta})\right]_{i+N} & =\frac{{\beta}_{i+N}}{\sqrt{{\beta}_{i}^{2}+{\beta}_{i+N}^{2}}}, i=1,2,\ldots,N_{2}.
		\end{align}
	\end{subequations}
	\begin{proposition}\label{PropoGradientss}
		The complex gradients $ \nabla_{\thetv}	\mathrm{SE}(\thetv,\betv) $ and  $\nabla_{\betv}	\mathrm{SE}(\thetv,\betv) $ are obtained in closed-forms by
		\begin{subequations}
			\begin{align}
				\nabla_{\thetv}	\mathrm{SE}(\thetv,\betv) &=[\nabla_{\thetv^{t}}	\mathrm{SE}(\thetv,\betv)^{\T}, \nabla_{\thetv^{r}}	\mathrm{SE}(\thetv,\betv)^{\T}]^{\T},\\
				\nabla_{\thetv^{t}}	\mathrm{SE}(\thetv,\betv)&=\frac{\tau_{\mathrm{c}}-\tau}{\tau_{\mathrm{c}}\log2}\sum_{k=1}^{K}\frac{	I_{k}\nabla_{\thetv^{t}}{S_{k}}-S_{k}	\nabla_{\thetv^{t}}{I_{k}}}{(1+\gamma_{k})I_{k}^{2}} ,\\
				\nabla_{\thetv^{r}}	\mathrm{SE}(\thetv,\betv)&=\frac{\tau_{\mathrm{c}}-\tau}{\tau_{\mathrm{c}}\log2}\sum_{k=1}^{K}\frac{	I_{k}\nabla_{\thetv^{r}}{S_{k}}-S_{k}	\nabla_{\thetv^{r}}{I_{k}}}{(1+\gamma_{k})I_{k}^{2}}, 
			\end{align}
		\end{subequations}
		where
		\begin{subequations}
			\begin{align}
				\nabla_{\thetv^{t}}S_{k}&=\begin{cases}
					\nu_{2k}\diag\bigl(\bB_{3t}\diag(\boldsymbol{{\beta}}^{t})\bigr)\\+\nu_{3k}\diag\bigl(\bB_{4t}\diag(\boldsymbol{{\beta}}^{t})\bigr) & w_{k}=t\\
					0 & w_{k}=r
				\end{cases}\label{derivtheta_t}\\
				\nabla_{\thetv^{r}}S_{k}&=\begin{cases}
					\nu_{2k}\diag\bigl(\bB_{3r}\diag(\boldsymbol{{\beta}}^{r})\bigr)\\
					+\nu_{3k}\diag\bigl(\bB_{4r}\diag(\boldsymbol{{\beta}}^{r})\bigr) & w_{k}=r\\
					0 & w_{k}=t
				\end{cases}\label{derivtheta_r}\\
				\nabla_{\thetv^{t}}I_k &=\diag\bigl(\tilde{\mathbf{A}}_{kt}\diag(\boldsymbol{{\beta}}^{t})\bigr)\label{derivtheta_t_Ik}\\
				\nabla_{\thetv^{r}}I_k &=\diag\bigl(\tilde{\mathbf{A}}_{kr}\diag(\boldsymbol{\beta}^{r})\bigr)\label{derivtheta_r_Ik}
			\end{align}
		\end{subequations}
		with 
		\begin{equation}
			\tilde{\mathbf{A}}_{kt}=\begin{cases}
				\bar{\nu}_{2k}\bB_{3}+\bar{\nu}_{2k}\bB_{4}+\sum\nolimits _{i\in\mathcal{K}_{t}}^{K}(\tilde{\nu}_{ki2}\bB_{3}+\tilde{\nu}_{ki2}\bB_{4}) & w_{k}=t\\
				\sum\nolimits _{i\in\mathcal{K}_{t}}(\tilde{\nu}_{ki2}\bB_{3}+\tilde{\nu}_{ki2}\bB_{4}) & w_{k}\neq t
			\end{cases}
		\end{equation}
		with  $ \bB_{3t}=\bR_{1}\bPhi_{2,t}\bR_{2} $, $ \bar{\bD} =\tr(\bB_{1}\bPhi_{1}^{\H})\bR_{t} $,  $ \bB_{4t}=\hat{\beta}_{2k}\tr(\bR_{t} \bPhi_{2,t}\bR_{2}) $ for $w_k\in\{t,r\}$, $ 	\nu_{2k}=2\tr(\bar{\boldsymbol{\Psi}}_{k})\tr\bigl(\bQ_{0k} \bR_{0k} 	\bar{\bD}-\bQ_{1k}\bR_{0k}^{2}\bQ_{1k}	\bar{\bD}+\bR_{0k}\bQ_{0k}	\bar{\bD}\bigr) $, $ \nu_{3k}=2\tr(\bar{\boldsymbol{\Psi}}_{k})\tr\bigl(\bQ_{2k}\bR_{2k}\bR_{2}-\bQ_{2k}\bR_{2k}^{2}\bQ_{2k}\bR_{2}+\bR_{2k}\bQ_{2k}\bR_{2}\bigr) $,
		$\bar{\nu}_{1k}=	\tr\bigl(\check{\bar{\boldsymbol{\Psi}}}_{k}\bar{\bD}\bigr)$,$\bar{\nu}_{2k}=\hat{\beta}_{1k}\tr\bigl(\check{\bar{\boldsymbol{\Psi}}}_{k}\bR_{2}\bigr)$, $\tilde{\nu}_{ki1}=\tr\bigl(\tilde{\mathbf{R}}_{ki}\bar{\bD}\bigr)$, and $\tilde{\nu}_{ki2}=\hat{\beta}_{1k}\tr\bigl(\tilde{\mathbf{R}}_{ki}\bR_{2}\bigr)$. In a similar way,
		the real-valued gradient $\nabla_{\betv}	\mathrm{SE}(\thetv,\betv) $ is given by
		\begin{subequations}\label{eq:deriv:wholebeta}
			\begin{align}
				\nabla_{\betv}	\mathrm{SE}(\thetv,\betv) &=[\nabla_{\betv^{t}}	\mathrm{SE}(\thetv,\betv)^{\T}, \nabla_{\betv^{r}}	\mathrm{SE}(\thetv,\betv)^{\T}]^{\T},\\
				\nabla_{\betv^{t}}	\mathrm{SE}(\thetv,\betv)&=\frac{\tau_{\mathrm{c}}-\tau}{\tau_{\mathrm{c}}\log2}\sum_{k=1}^{K}\frac{	I_{k}\nabla_{\betv^{t}}{S_{k}}-S_{k}	\nabla_{\betv^{t}}{I_{k}}}{(1+\gamma_{k})I_{k}^{2}}, \label{gradbetat:final}\\
				\nabla_{\betv^{r}}	\mathrm{SE}(\thetv,\betv)&=\frac{\tau_{\mathrm{c}}-\tau}{\tau_{\mathrm{c}}\log2}\sum_{k=1}^{K}\frac{	I_{k}\nabla_{\betv^{r}}{S_{k}}-S_{k}	\nabla_{\betv^{r}}{I_{k}}}{(1+\gamma_{k})I_{k}^{2}} ,
			\end{align}
		\end{subequations} 
		where
		\begin{subequations}
			\begin{align}
				\nabla_{\betv^{t}}S_{k}&=\begin{cases}
					2\real\{ \bar{\bD}\bigl(\diag\bigl(\bB_{3t}\herm\diag(\boldsymbol{{\beta}}^{t})\bigr)\bigr)\}\\+2\real\{\bR_{2}\bigl(\diag\bigl(\bB_{4t}\herm\diag(\boldsymbol{{\beta}}^{t})\bigr)\bigr)\} & w_{k}=t\\
					0 & w_{k}=r
				\end{cases}\label{derivbeta_t}\\
				\nabla_{\betv^{r}}S_{k}&=\begin{cases}
					2\real\{ \bar{\bD}\bigl(\diag\bigl(\bB_{3r}\herm\diag(\boldsymbol{{\beta}}^{r})\bigr)\bigr)\}\\+2\real\{\bR_{2}\bigl(\diag\bigl(\bB_{4r}\herm\diag(\boldsymbol{{\beta}}^{r})\bigr)\bigr)\} & w_{k}=r\\
					0 & w_{k}=t,
				\end{cases}\label{derivbeta_r}
			\end{align}
			\begin{align}
				\nabla_{\thetv^{t}}I_k
				&=2\real\{\diag\bigl(\tilde{\mathbf{A}}_{kt}\diag(\boldsymbol{\mathbf{\beta}}^{t})\bigr)\}\\
				\nabla_{\thetv^{r}}I_k
				&=2\real\{\diag\bigl(\tilde{\mathbf{A}}_{kr}\diag(\boldsymbol{\mathbf{\beta}}^{r})\bigr)\}.
			\end{align}
		\end{subequations}

	\end{proposition}
	\begin{proof}
		Please see Appendix~\ref{prop3}.	
	\end{proof}
	
	As mentioned, under independent Rayleigh fading conditions, which are unrealistic in practice, Algorithm \ref{Algoa1} would not execute  because no dependence on the phase shifts appears. Similarly, Algorithm  \ref{Algoa2} would execute only with respect to the amplitudes of STAR-RIS.
		\begin{remark}
	We have optimized the amplitude and phase shift,  separately, rather than optimizing them as a single complex although the presentation of the proposed method would be more elegant if we had chosen one variable for both of them. However,  despite that extensive numerical experiments revealed that both ways give the same performance in many cases, in some cases, the use of two separate variables yields a better performance. Hence, this numerical observation has lead to the separate independent optimization of the  amplitudes and phase shifts.
	\end{remark}
	\subsection{Complexity Analysis of Algorithms \ref{Algoa1} and \ref{Algoa2}}	
	Herein, we provide the complexity analysis for each iteration of Algorithm \ref{Algoa1}. Regarding $ \bR_{0k} $, $ \bR_{1k} $, $ \bR_{2k} $, the traces require  $\mathcal{O}(N^2)$ complex multiplications because $ \bPhi_{1} $ is diagonal, while their overall expression requires in total  $\mathcal{O}(M^3+N^2)$ complex multiplications since  the trace is multiplied with a $ M\times M $ matrix. Also, $ \bar{\bPsi}_{k} $, including a matrix inversion, takes 
	$\mathcal{O}(M^3)$  complex multiplications, which could be reduced by applying eigenvalue decomposition (EVD). Thus, $ \mathrm{SE} $ requires $\mathcal{O}(K(N^2+M^3))$ iterations. The same number of iterations is required by $ \nabla_{\bar{\thetv}}\mathrm{SE}(\bar{\thetv}) $. The complexity of Algorithm \ref{Algoa2} can be shown that is the same by following the same steps.
	
	\subsubsection*{Convergence Analysis of Algorithm 2}
		We follow  standard arguments for projected gradient methods to show the convergence of Algorithm 2. First, the gradients $\nabla_{\thetv}f(\thetv,\betv)$ and $\nabla_{\betv}f(\thetv,\betv)$ are Lipschitz continuous\footnote{ A function $\bh(\bx)   $ is  Lipschitz continuous over the set $D$ if there exists $L>0$ such that $||\bh(\bx)-\bh(\by)  ||\leq L||\bx-\by||_2$} over the feasible set since they comprise basic functions. If we denote $L_{\thetv }$ and $L_{\betv}$  the Lipschitz constant of $\nabla_{\thetv}f(\thetv,\betv)$ and $\nabla_{\betv}f(\thetv,\betv)$, respectively, we have that \cite[Chapter 2]{Bertsekas1999}
		\begin{align}
			f(\bx,\by) &\geq f(\thetv,\betv)
			+\langle	\nabla_{\thetv}f(\thetv,\betv),\bx-\thetv\rangle-\frac{1}{L_{\thetv }}\|\bx-\thetv\|^{2}_{2}
			\\\nn
			&\quad\quad+\langle\nabla_{\betv}f(\thetv,\betv),\by-\betv\rangle-\frac{1}{L_{\betv}}\|\by-\betv\|^{2}_{2}\\\nn
			&\geq f(\thetv,\betv)
			+\langle	\nabla_{\thetv}f(\thetv,\betv),\bx-\thetv\rangle-\frac{1}{L_{\max }}\|\bx-\thetv\|^{2}_{2}
			\\ \nn&\quad\quad+\langle\nabla_{\betv}f(\thetv,\betv),\by-\betv\rangle-\frac{1}{L_{\max}}\|\by-\betv\|^{2}_{2}		
		\end{align}
		where $L_{\max}=\max(L_{\thetv },L_{\betv})$. 
		Hence, the termination of the line search procedure of Algorithm 2  is achieved in finite iterations because the condition in Step \ref{ls:end} must be satisfied when  $\mu_n <L_{\max}$. 
		In particular, given $\mu_{n-1}$, the maximum number of steps in the line search procedure is $\left\lceil \frac{\log(L_{\max}\mu_{n-1})}{\log\kappa}\right\rceil $, where $\log()$ denotes the natural logarithm and  		$\left\lceil \cdot \right\rceil$  		denotes the smallest integer that is larger than or equal to the argument. 
		Moreover, because of the line search, we automatically obtain  an increasing sequence of objectives, i.e., $f(\thetv^{n+1},\betv^{n+1})\geq f(\thetv^{n},\betv^{n})$. Given that the feasible sets $\Theta$  and $\mathcal{B}$ are compact, $f(\thetv^{n},\betv^{n})$ must converge.  Notably,  $L_{\thetv }$ and $L_{\betv}$ are not required to run Algorithm 2.
	\section{Numerical Results}\label{numerical11}
	In this section, we discuss the numerical results corresponding to the downlink sum SE of double RIS-assisted mMIMO systems under correlated Rayleigh fading conditions. MC simulations corroborate our analysis even for finite (conventional) system dimensions, which agrees with a similar observation in~\cite{Couillet2011,Hoydis2013,Papazafeiropoulos2015a}.
	
	The simulation setup consists of a conventional RIS and a STAR that facilitate the communication between an mMIMO BS, which serves $ K = 4 $ UEs. Both surfaces are formed by a UPA of $ N=64 $ elements with each having dimensions $ d_{\mathrm{H}}\!=\!d_{\mathrm{V}}\!=\!\lambda/4 $, while the BS is formed by a uniform linear array (ULA) of $ M =64$ antennas. The  3D spatial locations of the network nodes being the  BS, conventional RIS, and STAR-RIS are given as $(x_B,~ y_B,~ z_B) = (0,~0,~0)$, $(x_R,~ y_R,~ z_R)=(50,~ 10, 20)$, and $(x_{SR},~ y_{SR},~z_{SR})=(100,~ 30,  20)$ respectively, all in meter units. Regarding the UEs in the $r$ region, they are located on a straight line between $(x_{SR}-\frac{1}{2}d_0,~y_{SR}-\frac{1}{2}d_0)$ and $(x_{SR}+\frac{1}{2}d_o,~y_{SR}-\frac{1}{2}d_0)$ with equal distances between each two adjacent users, and $d_0 = 20$~m in our simulations. UEs in the $t$ region are located between $(x_{SR}-\frac{1}{2}d_0,~y_{SR}+\frac{1}{2}d_0)$ and $(x_{SR}+\frac{1}{2}d_o,~y_{SR}+\frac{1}{2}d_0)$. In this work, we consider distance-based path-loss, where the channel gain of a given link $j$ is $\tilde \beta_j = A d_j^{-\alpha_{j}}$ with the channel gain of a given link $j$ being $\tilde \beta_j = A d_j^{-\alpha_{j}}$. The variables $A$ and  $\alpha_{j}$ denote the area of each reflecting element at the RIS and the path-loss exponent, respectively.  The path losses $ \tilde \beta_j $, $ \tilde{ \beta}_{g} $,  and $ \bar{ \beta}_{k} $ are assumed to have the same values but the latter is assumed to have a further penetration loss equal to $ 15 $ dB. The correlation matrices $ \bR_{\mathrm{BS}}$ and $\bR_{\mathrm{RIS}} $ are  evaluated based on \cite{Hoydis2013} and \cite{Bjoernson2020}, respectively. We assume that both surfaces have the same correlation. The variance of the noise is $ \sigma^2=-174+10\log_{10}B_{\mathrm{c}} $, where $B_{\mathrm{c}}=200~\mathrm{kHz}$ is the bandwidth. The coherence time is $ T_{c} = 1~\mathrm{ms} $, i.e., each coherence block consists of $ \tau_{c} = 200 $ samples, and we assume that the duration of the channel estimation phase is $ \tau = 20 $ samples.

	As  baseline schemes, we consider the conventional double-RIS, where the STAR-RIS is replaced by a conventional surface split into two subsurfaces that  consist of transmitting-only and reflecting-only elements, each with $ N_{2t} $ and $ N_{2r} $ elements, such that $ N_{2t}+N_{2r} =N_{2}$.  Also, we consider the single-RIS counterparts, which consist of a single STAR-RIS with $ N=N_{1}+N_{2} $ elements, or a single conventional RIS with the same number of elements.

	Fig. \ref{fig2} depicts the achievable sum SE versus the total number of  RIS elements $ N=N_{1}+N_{2} $ by studying the  effect of spatial correlation while changing the size of each RIS element. Obviously, the downlink sum SE increases with $ N $ as expected. Regarding the effect of spatial correlation, it is shown that the performance increases as  the correlation decreases by increasing the size  of the RIS elements. In addition, we have depicted the scenario of random phase shifts, i.e., no optimization has taken place on any of the surfaces. Hence, the performance is lower than in the case of optimized RISs. Also, we have provided the performance of no direct signal, which shows that the surfaces contribute to the sum-rate. In the same figure, we have depicted the impact of imperfect CSI by varying $ \tau $, which expresses the duration  of the channel estimation phase. In particular, the case $ \tau =0$ corresponds to the perfect CSI scenario, which presents the best performance. The case, where $ \tau =20$ samples presents lower performance, while by increasing $ \tau $ to $40  $ samples, the performance worsens even more.
	
	\begin{figure}[!h]
		\begin{center}
			\includegraphics[width=0.8\linewidth]{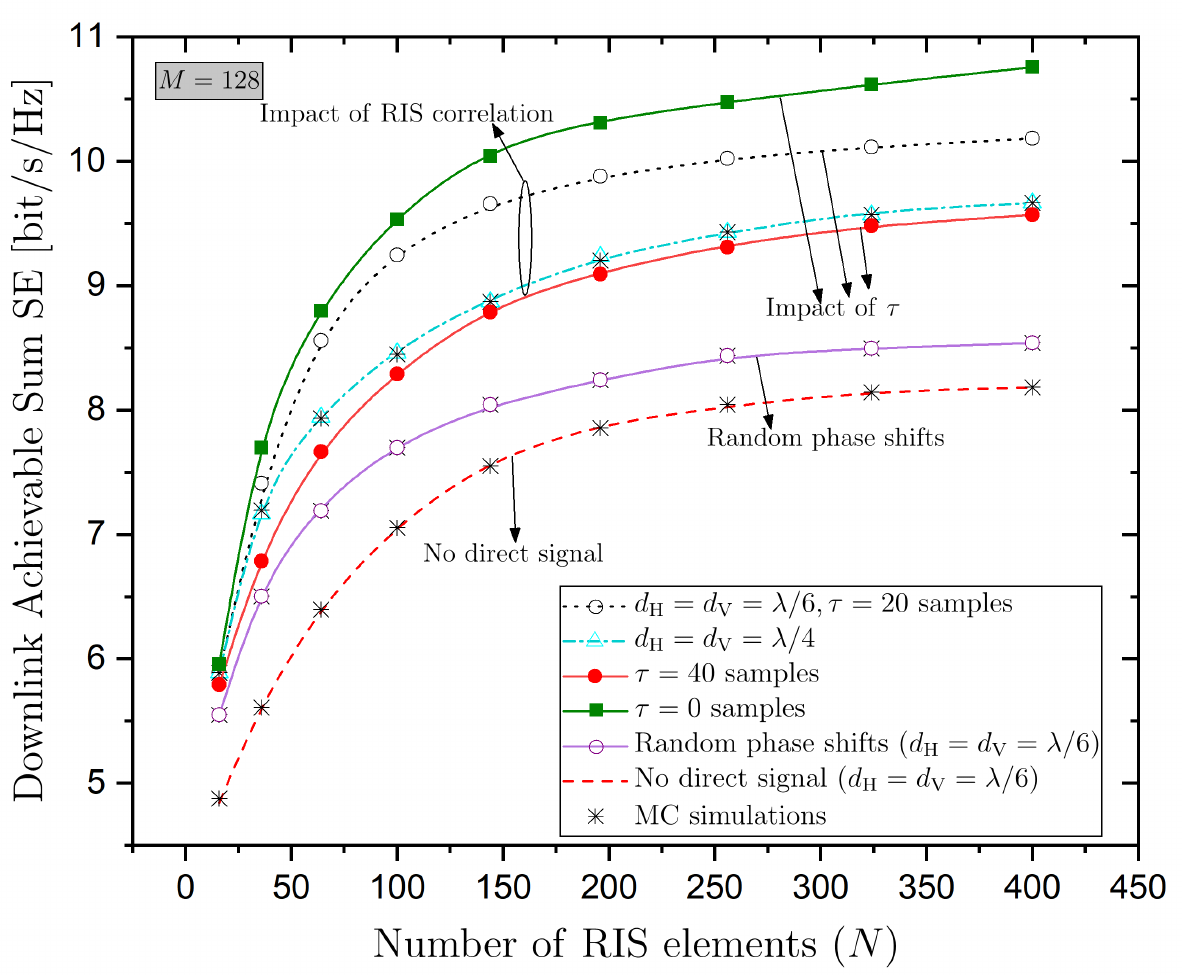}
			\caption{{Downlink achievable sum SE versus the total number of RIS elements antennas $N$ of a RIS/STAR-RIS assisted mMIMO system with imperfect CSI ($ M=100 $, $ K=4 $) for varying conditions (Analytical results and MC simulations). }}
			\label{fig2}
		\end{center}
	\end{figure}

	Fig. \ref{fig21} demonstrates the superiority of the proposed combination of RIS with a STAR-RIS in terms of the achievable sum SE versus the total number of  RIS elements $ N$. Specifically, we observe that the RIS/STAR-RIS outperforms the conventional double-RIS design including the double and single reflected links. The reason is the introduction of the advantageous STAR-RIS, which exploits more degrees of freedom with respect to a conventional reflecting-only RIS. In comparison to single-RIS counterparts, the double-RIS design enjoys the cooperative PB gain, and additionally, to this, the proposed model outperforms the conventional single RIS due to the presence of more adjustable (reflection and transmission) parameters.
	
	\begin{figure}[!h]
		\begin{center}
			\includegraphics[width=0.8\linewidth]{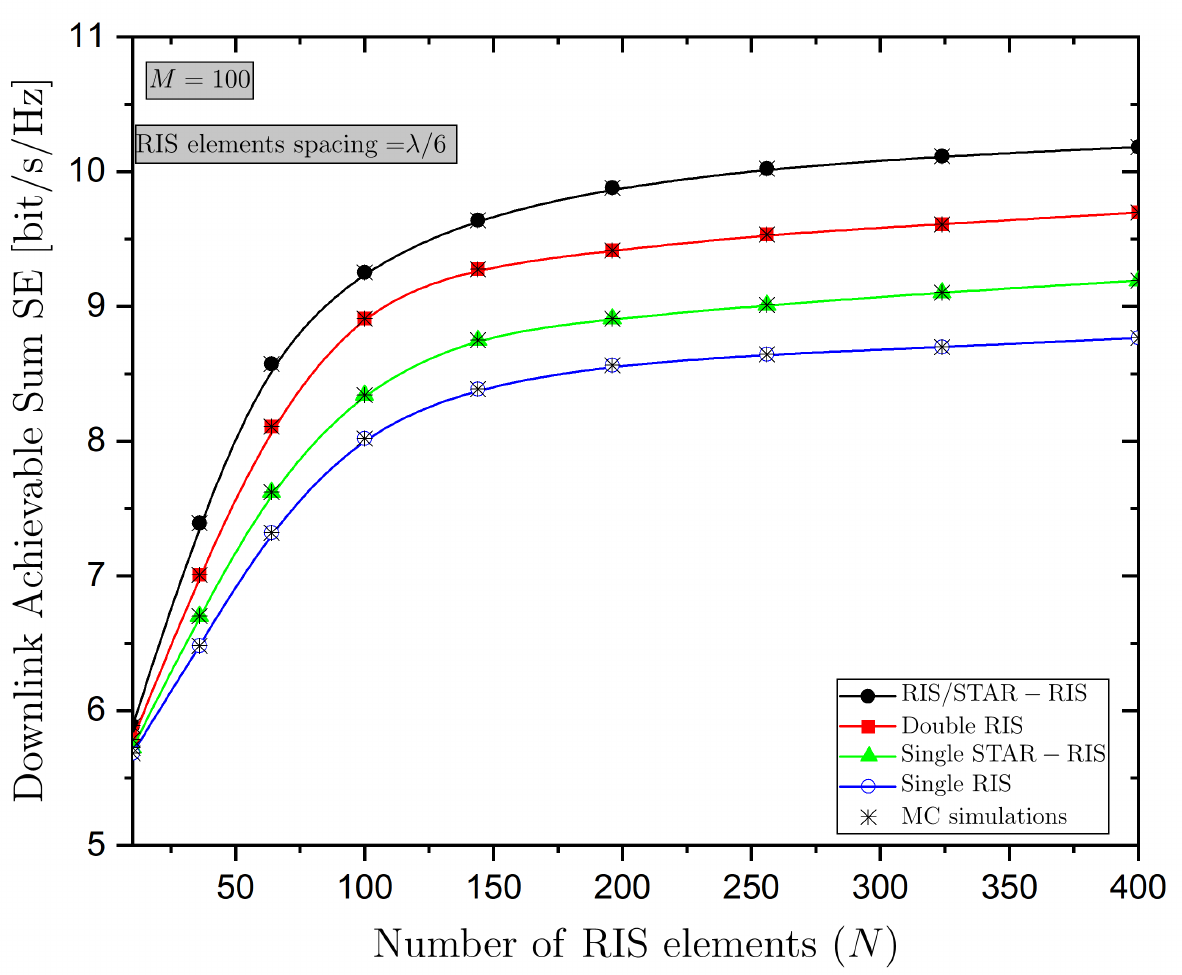}
			\caption{{Downlink achievable sum SE versus the total number of RIS elements antennas $N$ of a RIS/STAR-RIS assisted mMIMO system with imperfect CSI ($ M=100 $, $ K=4 $) for different architectures (Analytical results and MC simulations). }}
			\label{fig21}
		\end{center}
	\end{figure}
	
	Figs. \ref{fig3} illustrates the achievable sum SE versus the number of BS antennas $ M $ for different architectures being the RIS/STAR-RIS, double-RIS, single STAR-RIS, and single RIS. Apart from the fact that  the sum SE  increases with $ M $ in all cases. In particular, we observe that when two RIS are deployed but the second surface is a STAR-RIS, the sum SE is higher. In the case of a single RIS, despite the type of the surface, which can be a conventional RIS or a STAR-RIS, the performance is lower than having two RISs as expected because this scenario does not enjoy the double-RIS PB gain. In all cases, which include a STAR-RIS, the performance is better than the reflective only counterpart since both the transmission and reflection coefficients of each element can be optimized.
	
	\begin{figure}%
		\centering
		\includegraphics[width=0.8\linewidth]{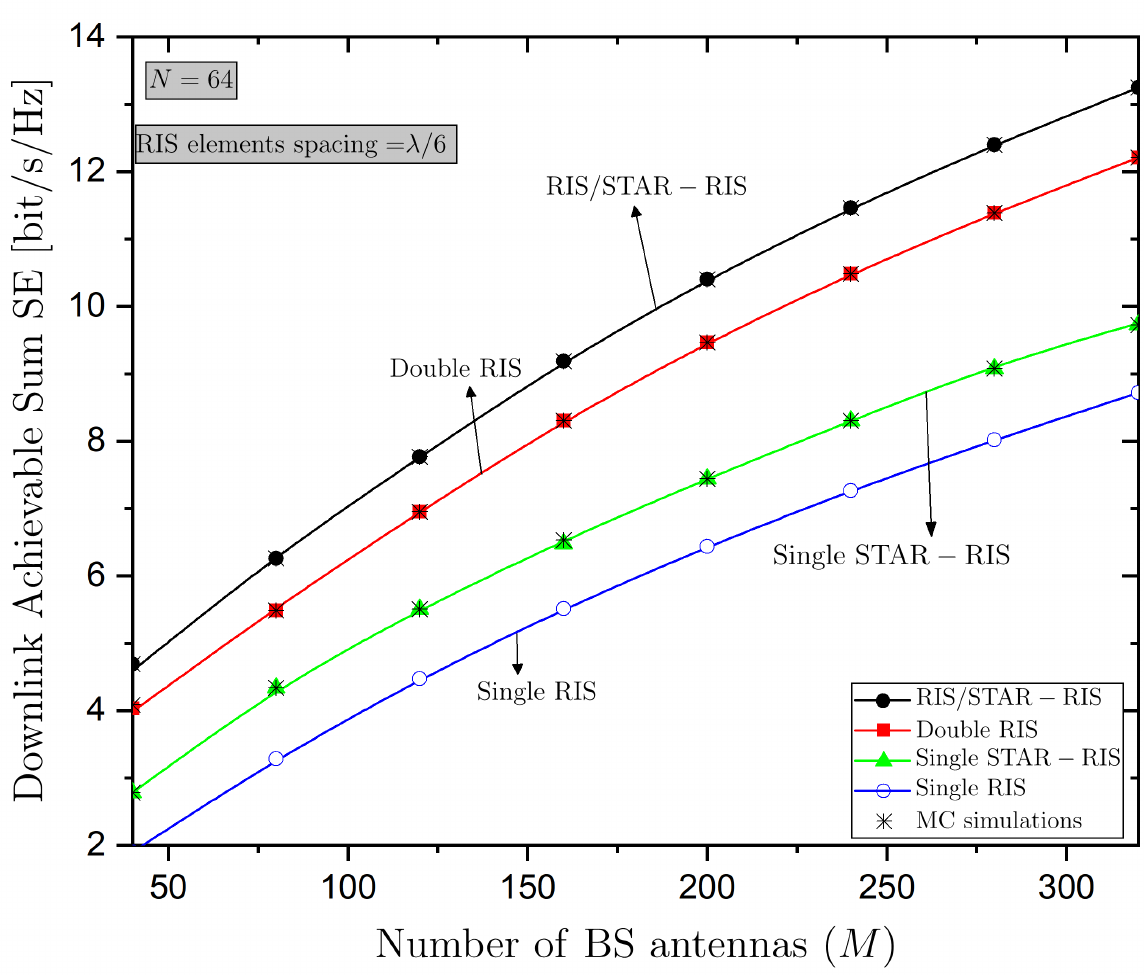}
		\caption{ Downlink achievable sum SE versus the number of BS antennas $M$ of a RIS/STAR-RIS assisted mMIMO system with imperfect CSI for $N=64 $, $ K=4 $ under varying conditions (Analytical results). }
		\label{fig3}
	\end{figure}

	Fig. \ref{fig4} shows  the achievable sum SE versus the SNR for different layouts as in Fig. \ref{fig21}. The RIS/STAR-RIS implementation presents the best performance by exploiting its STAR-RIS part, which brings a higher degree of flexibility by optimizing both the transmission and reflection coefficients of each element. On the contrary, in the case of conventional RIS, only one type of the coefficients can be optimized each time, i.e., the transmission or the reflection coefficient. Also, the RIS/STAR-RIS and double RIS layouts, which include two surfaces, benefit from the  double-RIS PB gain. 
	
	\begin{figure}[!h]
		\begin{center}
			\includegraphics[width=0.8\linewidth]{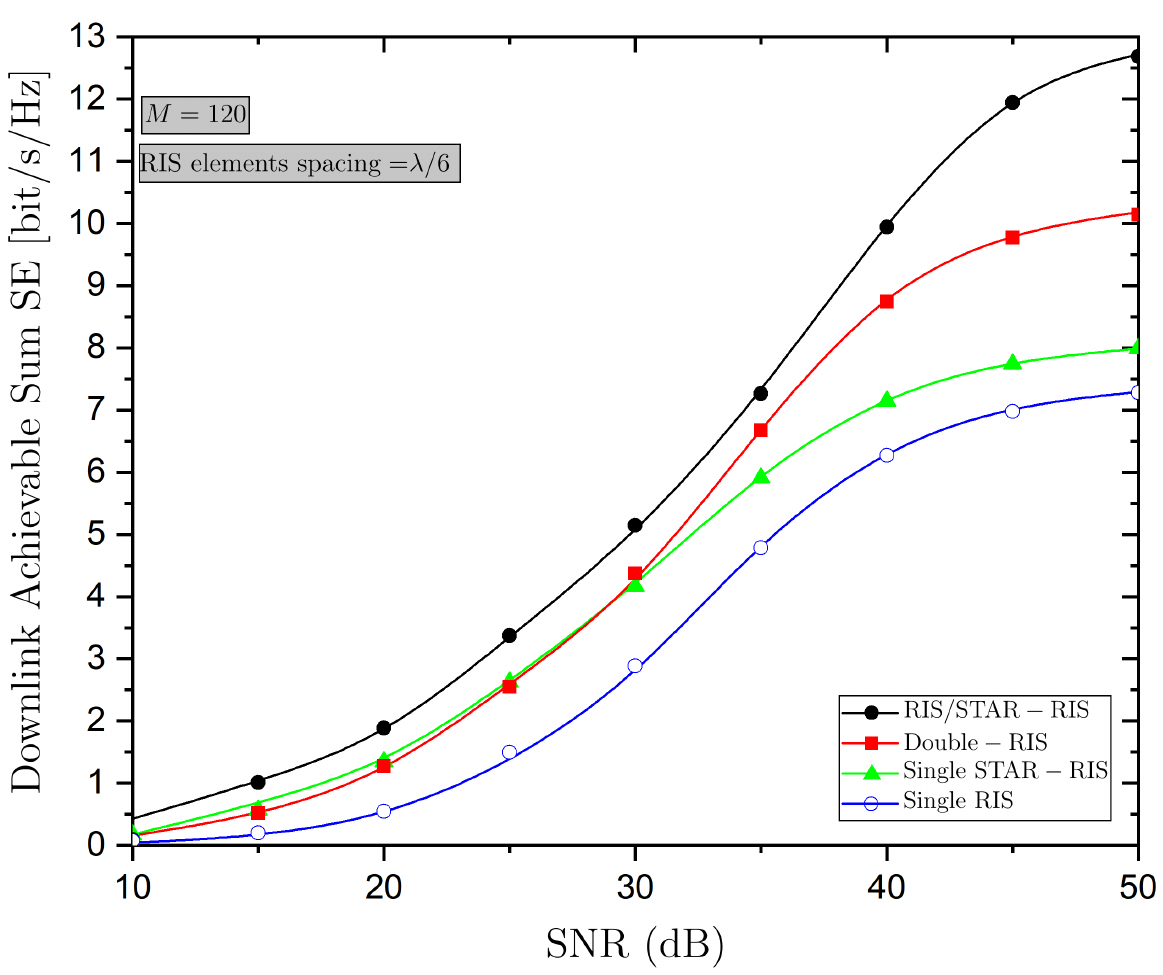}
			\caption{{Downlink achievable sum SE versus the SNR of a RIS/STAR-RIS assisted mMIMO system with imperfect CSI ($M=100$, $ N=64 $, $ K=4 $) for varying conditions (Analytical results). }}
			\label{fig4}
		\end{center} 
	\end{figure}

		Fig. \ref{fig6} illustrates  the achievable sum SE versus the number of elements $ N_{1} $ of surface 1, i.e., the conventional RIS given the total number of elements $ N=N_{1}+N_{2}=320$. For the sake of reference, we have depicted the scenarios of double-RIS, single RIS, and STAR/RIS. In other words, we have shown the scenarios with 2 surfaces and their single-surface counterparts. The former always achieve better rate performance compared to the single RIS baselines. Also, the the RIS/STAR-RIS model performs better than the reflective-only double-RIS model. Moreover, we observe that the rate is maximized when the 2 surfaces have almost equal number of elements.
	
		\begin{figure}[!h]
		\begin{center}
			\includegraphics[width=0.8\linewidth]{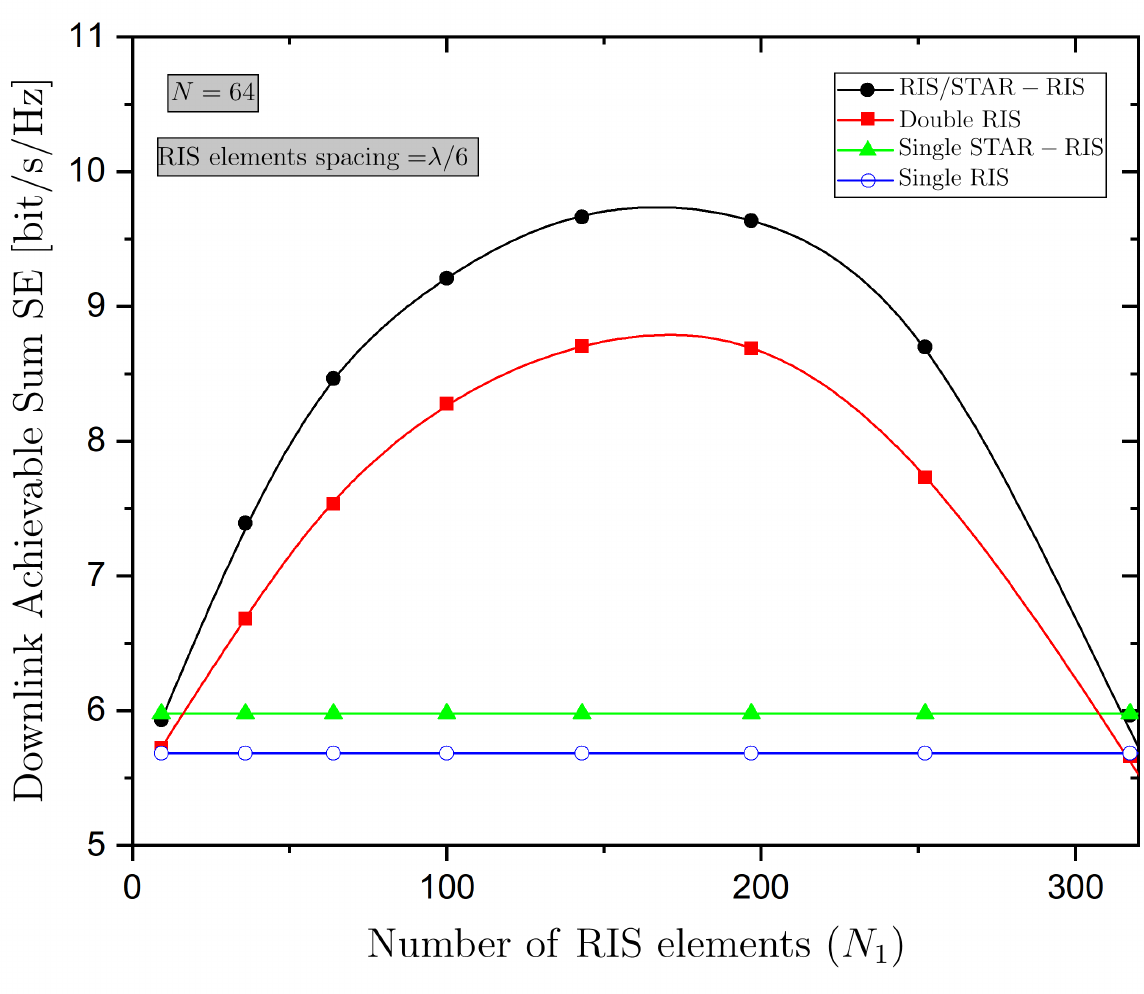}
			\caption{{Downlink achievable sum SE versus the SNR of a RIS/STAR-RIS assisted mMIMO system with imperfect CSI ($M=100$, $ N=64 $, $ K=4 $) for varying conditions (Analytical results). }}
			\label{fig6}
		\end{center} 
	\end{figure}
	
	Fig. \ref{fig5} elaborates on the convergence of the proposed projected gradient algorithm for the STAR-RIS. Given that \ref{Algoa1} and \ref{Algoa2} have similar structures but  \ref{Algoa1} is simpler, similar observations concern it.  In particular, in the case of  \ref{Algoa2}, we depict the  achievable sum SE against the iteration count returned  for 5 different randomly generated initial points. By assuming equal power splitting between transmission and reception mode, the initial values for $\betv_{r}$ and $\betv_{t}$ are $\sqrt{0.5}$. The initial values for $\thetv_{r}$ and  $\thetv_{r}$ are drawn from the Uniform distribution over $[0,2\pi]$. The algorithm terminates when the  increase of the objective between the two last iterations is less than $10^{-5}$ or the number of iterations is larger than $200$. Since the problem is nonconvex, Algorithm \ref{Algoa2} provides a not necessarily optimal solution. Hence, Algorithm \ref{Algoa2} may start from different initial points and lead to different points with different convergence rates. To address this sensitivity issue, we run Algorithm \ref{Algoa2} from different initial points and select the best convergent solutions, simulations have shown that to achieve a good trade-off between complexity and  achievable sum SE, a good option option is to run  \ref{Algoa2} from 5 randomly generated initial points

	\begin{figure}[!h]
		\begin{center}
			\includegraphics[width=0.8\linewidth]{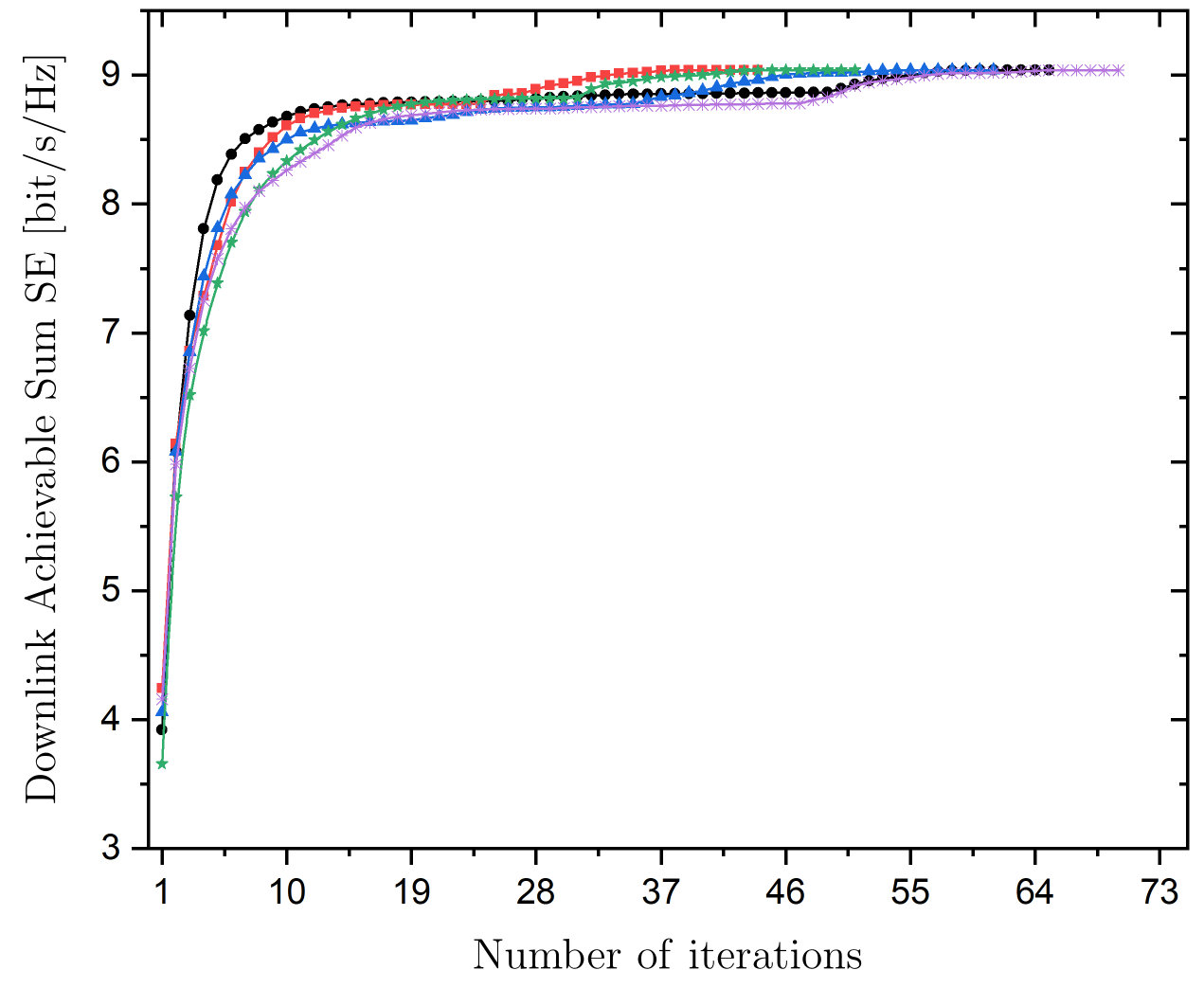}
			\caption{{Convergence of Algorithm \ref{Algoa2} for a RIS/STAR-RIS assisted mMIMO system with imperfect CSI ($M=100$, $ N=64 $, $ K=4 $). }}
			\label{fig5}
		\end{center} 
	\end{figure}

	\section{Conclusion} \label{Conclusion} 
	In this paper, we proposed a RIS/STAR-RIS assisted mMIMO communication system by exploiting the cooperative beamforming  and full coverage under correlated Rayleigh fading conditions and the coexistence of both double and single links. Moreover, we obtained the estimated channels and derived the DE of the sum SE in closed form in terms of large-scale statistics, which induce low overhead. Next, we formulated and solved the cooperative beamforming optimization problem to maximise the sum SE. The closed-form gradients came with low complexity and reduced overhead which is crucial, especially, in STAR-RIS architectures that include double parameters. One of the two main reasons for reduced overhead during optimization is that it can be performed at every several coherence intervals. The other reason is that we optimized the amplitudes and phase shifts of the STAR-RIS simultaneously. Simulation results showed substantial performance improvement compared to conventional double RIS or single RIS counterparts. As a future work, the study of more general multi-RIS architectures  with more than two hops for reflection and  transmission could be investigated to reach the full potential of a smart wireless environment.

	\begin{appendices}
		\section{Proof of Lemma~\ref{PropositionDirectChannel}}\label{lem1}
		The LMMSE estimator of $ \bh_{k} $ results from the minimization of $ \tr\!\big(\EE\big\{\!(\hat{\bh}_{k}-{\bh}_{k})(\hat{\bh}_{k}-{\bh}_{k})^{\H}\!\big\}\!\big) $, which gives
		\begin{align}
			\hat{\bh}_{k} =\EE\!\left\{\br_{k}\bh_{k}^{\H}\right\}\left(\EE\!\left\{\br_{k}\br_{k}^{\H}\right\}\right)^{-1}\br_{k}.\label{Cor6}
		\end{align}
		
		The first expectation is obtained by exploiting  that the channel and the receiver noise are uncorrelated. In particular, we have
		\begin{align}
			\EE\left\{\br_{k}\bh_{k}^{\H}\right\}
			&=\EE\left\{\bh_{k}\bh_{k}^{\H}\right\}=\bR_{0k}.\label{Cor0}
		\end{align}
		
		The second term in \eqref{Cor6} is obtained as
		\begin{align}
			\EE\left\{\br_{k}\br_{k}^{\H}\right\}&=\bR_{0k} +\frac{\sigma^2}{ \tau P }\Id_{M}.\label{Cor1}
		\end{align}
		The LMMSE estimate in \eqref{estim1} is obtained by substituting \eqref{Cor0} and \eqref{Cor1} into \eqref{Cor6}. Also, the covariance matrix of the estimated channel is obtained as
		\begin{align}
			\EE\left\{\hat{\bh}_{k}	\hat{\bh}_{k}^{\H}\right\}=\bR_{0k}\bQ_{0k}\bR_{0k}.\label{var1}
		\end{align}

		\section{Proof of Proposition~\ref{Proposition:DLSINR}}\label{Proposition1}
		
		The DE of  $ S_k $ in \eqref{Sig} can be written as 
		\begin{align}
			S_{k}&=|\EE\{\bar{\bh}_{k}^{\H}\hat{\bar{\bh}}_{k}\}|^{2}\asymp|\EE\{\hat{\bar{\bh}}_{k}^{\H}\hat{\bar{\bh}}_{k}\}|^{2}\label{term1} \\
			&=|\tr\left(\bar{\bPsi}_{k}\right)\!|^{2}\label{term2},
		\end{align}
		where, in \eqref{term1}, we have taken into account the independence between the  channel and its estimated version. The last equation is obtained based on \cite[Lem. 4]{Papazafeiropoulos2015a}.

		Regarding the first term of $ I_k $ in \eqref{Int}, we have
		\begin{align}
			&\!\!\EE\big\{ \big| \bar{\bh}_{k}^{\H}\hat{\bar{\bh}}_{k}-\EE\big\{
			\bar{\bh}_{k}^{\H}\hat{\bar{\bh}}_{k}\big\}\big|^{2}\big\}\!=\!
			\EE\big\{ \big| \bar{\bh}_{k}^{\H}\hat{\bar{\bh}}_{k}\big|^{2}\big\}\!-\!\big|\EE\big\{
			\bar{{\bh}}_{k}^{\H}\hat{\bar{\bh}}_{k}\big\}\big|^{2} \label{est2}\\
			&\asymp\EE\big\{ \big| \hat{\bar{\bh}}_{k}^{\H} \hat{\bar{\bh}}_{k} +|\bar{\bee}_{k}^{\H}\hat{\bar{\bh}}_{k}\big|^{2}\big\}-\big|\EE\big\{
			\hat{\bar{\bh}}_{k}^{\H}\hat{\bar{\bh}}_{k}\big\}\big|^{2}\label{est3} \\
			&=\EE\big\{|\bar{\bee}_{k}^{\H}\hat{\bar{\bh}}_{k}|^{2}\big\} \label{est5}\\
			&\asymp\tr\!\left( \bar{\bR}_{k}\bar{\bPsi}_{k}\right)-\tr\left(\bar{ \bPsi}_{k}^{2}\right),\label{est4}
		\end{align}
		where in~\eqref{est3}, we have used \eqref{estimatedchannel}, and that $\EE\left\{ |X+Y|^{2}\right\} =\EE\left\{ |X|^{2}\right\} +\EE\left\{ |Y^{2}|\right\}$, which is valid for any two uncorrelated random variables while one of them has zero mean value. The last equation is derived by taking into account the uncorrelation between the two random vectors and by application of \cite[Lem. 4]{Papazafeiropoulos2015a}.
		
		The second term of $ I_k $ in \eqref{Int} is obtained as
		\begin{align}
			&\EE\big\{ \big| \bar{\bh}_{k}^{\H}\hat{\bar{\bh}}_{i}\big|^{2}\big\}\asymp\tr\!\left(\bar{\bR}_{k}\bar{\bPsi}_{i} \right)\label{54}
		\end{align}
		due to the   uncorrelation between $ {\bh}_{k} $ and $ \hat{\bh}_{i} $ and \cite[Lem. 4]{Papazafeiropoulos2015a}.  The normalization parameter is obtained  as
		\begin{align}
			\!\!\!	\lambda=\frac{1}{\sum_{i=1}^{K}\!\EE\{\mathbf{f}_{i}^{\H}\mathbf{f}_{i}\}}=\frac{1}{\sum_{i=1}^{K}\!\mathbb{E}\{\hat{\bar{\mathbf{h}}}_{i}^{\H}\hat{\bar{\mathbf{h}}}_{i}\}}\asymp\frac{1}{\sum_{i=1}^{K}\!\tr(\bar{\boldsymbol{\Psi}}_{i})}.\label{normalization}
		\end{align}
		The proof concludes by inserting \eqref{term2},  \eqref{est4},  \eqref{54}, and \eqref{normalization} into \eqref{Sig} and \eqref{Int}.
		
		\section{Proof of Proposition~\ref{PropGradients}}\label{Prop2}
		Based on \eqref{LowerBound}, we can easily obtain
		\begin{align}
			\nabla_{\bar{\thetv}}\mathrm{SE}(\bar{\thetv})&=\frac{\tau_{\mathrm{c}}-\tau}{\tau_{\mathrm{c}}\log2}\sum_{k=1}^{K}\frac{	I_{k}\nabla_{\bar{\thetv}}{S_{k}}-S_{k}	\nabla_{\bar{\thetv}}{I_{k}}}{(1+\gamma_{k})I_{k}^{2}}.
		\end{align}
		For the derivation of $ \nabla_{\bar{\thetv}}{S_{k}} $ with fixed $ \bPhi_{2,w_{k}} $, we 
		derive its 
		differential  
		as
		\begin{align}
			d(S_{k})&=d\bigl(\tr(\bar{\boldsymbol{\Psi}}_{k})^{2}\bigr)\nn\\
			&=2\tr(\bar{\boldsymbol{\Psi}}_{k})d\tr(\boldsymbol{\Psi}_{k})\nn\\
			&=2\tr(\bar{\boldsymbol{\Psi}}_{k})\tr(d(\bar{\boldsymbol{\Psi}}_{k})).\label{eq:dSk}
		\end{align}
		
		The next step concerns the derivation of $ d(\boldsymbol{\Psi}_{k}) $. First, we write $ \bar{\bPsi}_{k} $ as
		\begin{align}
			\bar{\bPsi}_{k}&=\bPsi_{k}+\bPsi_{1k}+\bPsi_{2}\\
			&=\bR_{0k}\bQ_{0k}\bR_{0k}+\bR_{1k}\bQ_{1k}\bR_{1k}+\bR_{2k}\bQ_{2k}\bR_{2k}.\label{psiK}
		\end{align}
		By using \cite[Eq. (3.35)]{hjorungnes:2011}, its  differential becomes
		\begin{align}
			d(\boldsymbol{\bar{\Psi}}_{k})&=d(\bR_{0k})\bQ_{0k}\bR_{0k}+\bR_{0k}d(\bQ_{0k})\bR_{0k}\nn\\
			&+\bR_{0k}\bQ_{0k}d(\bR_{0k})
			+d(\bR_{1k})\bQ_{1k}\bR_{1k}\nn\\
			&+\bR_{1k}d(\bQ_{1k})\bR_{1k}+\bR_{1k}\bQ_{1k}d(\bR_{1k})\label{eq:dPsik}
		\end{align}
		since $ \bR_{2k} $ and $  \bQ_{2k}$ do not depend on $ \bPhi_{1}  $.
		
		Regarding $ d(\bQ_{0k}) $, use of  \cite[eqn. (3.40)]{hjorungnes:2011} gives 
		\begin{align}
		&	d(\bQ_{0k})  =d\bigl(\mathbf{R}_{0k}+\frac{\sigma}{\tau P}\mathbf{I}_{M}\bigr)^{-1}\nn\\
			&=-\bigl(\mathbf{R}_{0k}+\frac{\sigma^{2}}{\tau P}\mathbf{I}_{M}\bigr)^{-1}d\bigl(\mathbf{R}_{0k}+\frac{\sigma^{2}}{\tau P}\mathbf{I}_{M}\bigr)\bigl(\mathbf{R}_{0k}+\frac{\sigma}{\tau P}\mathbf{I}_{M}\bigr)^{-1}\nonumber \\
			& =-\bQ_{0k}d(\mathbf{R}_{0k})\bQ_{0k}.\label{eq:dQ0k}
		\end{align}
		Similarly, we obtain 
		\begin{align}
			d(\bQ_{1k}) &=-\bQ_{1k}d(\mathbf{R}_{1k})\bQ_{1k}.\label{eq:dQ1k}	\end{align}
		
		Inserting \eqref{eq:dQ0k} and \eqref{eq:dQ1k} into \eqref{eq:dPsik}, it yields
		\begin{align}
		&	d(\boldsymbol{\Psi}_{k})=d(\bR_{0k})\bQ_{0k}\bR_{0k}-\bR_{0k}\bQ_{1k}d(\mathbf{R}_{1k})\bQ_{1k}\bR_{0k})\nn\\
			&+\bR_{0k}\bQ_{0k}d(\bR_{0k}+d(\bR_{1k})\bQ_{1k}\bR_{1k}\nn\\
			&-\bR_{1k}\bQ_{1k}d(\mathbf{R}_{1k})\bQ_{1k}\bR_{1k}+\bR_{1k}\bQ_{1k}d(\bR_{1k}).\label{eq:dPsik1}
		\end{align}

		Now, we have to obtain $ d({\mathbf{R}}_{0k}) $ and  $ d({\mathbf{R}}_{1k}) $ given that $ \bR_{0k}=\hat{\beta}_{k}\tr(\bR_{1}\bPhi_{2,w_{k}}\bR_{2}\bPhi_{2,w_{k}}^{\H})\tr(\bB_{1}\bPhi_{1}^{\H})\bR_{t} $ and $ \bR_{1k}=\hat{\beta}_{1k}\tr(\bB_{2}\bPhi_{1}^{\H})\bR_{1} $, where $ \bB_{1}=\bR_{1}\bPhi_{1}\bR_{2} $ and $ \bB_{2}=\bR_{t} \bPhi_{1}\bR_{1} $. First, we write $ d({\mathbf{R}}_{0k}) $ as
		\begin{align}
		&	d({\mathbf{R}}_{0k})=\bar{C}	\tr(\bB_{1}^{\H}\bPhi_{1}+\bB_{1}d(\bPhi_{1}^{\H}))\bR_{t}\\
			&=\bar{C}	\bR_{t}\bigl(\bigl(\diag\bigl(\bB_{1}\herm\bigr)\trans d(\boldsymbol{\bar{\theta}})+\bigl(\diag\bigl(\bB_{1}\bigr)\bigr)\trans d(\boldsymbol{\bar{\theta}}^{\ast})\bigr),\label{r0k}
		\end{align}
		where $\bar{C}= \hat{\beta}_{k}\tr(\bR_{1}\bPhi_{2,w_{k}}\bR_{2}\bPhi_{2,w_{k}}^{\H}) $, while the last equation is obtained by exploiting that  $ \bPhi_{1} $ is diagonal. 
		
		Similarly, we have
		\begin{align}
		&	d({\mathbf{R}}_{1k})=\hat{\beta}_{1k}\tr(\bB_{2}^{\H}d(\bPhi_{1})+\bB_{2}d(\bPhi_{1}))\bR_{1}\\
			&=\hat{\beta}_{1k}\bR_{1}\bigl(\bigl(\diag\bigl(\bB_{2}\herm\bigr)\trans d(\boldsymbol{\bar{\theta}})+\bigl(\diag\bigl(\bB_{2}\bigr)\bigr)\trans d(\boldsymbol{\bar{\theta}}^{\ast})\bigr).\label{r1k}
		\end{align}

		Substitution of  \eqref{eq:dPsik1} into \eqref{eq:dSk} results in
		\begin{align}
			d(S_{k})&=2\tr(\bar{\boldsymbol{\Psi}}_{k})\tr\bigl(\bQ_{0k} \bR_{0k} d(\bR_{0k})-\bQ_{1k}\bR_{0k}^{2}\bQ_{1k}d(\mathbf{R}_{0k})\nn\\
			&+			\bR_{0k}\bQ_{0k}d(\bR_{0k})+\bQ_{1k}\bR_{1k}d(\bR_{1k})\nn\\
			&	-\bQ_{1k}\bR_{1k}^{2}\bQ_{1k}d(\mathbf{R}_{1k})	\bR_{1k}\bQ_{1k}d(\bR_{1k})\bigr)\\
			&=\nu_{0k}\bigl(\bigl(\diag\bigl(\bB_{1}\herm\bigr)\trans d(\boldsymbol{\bar{\theta}})+\bigl(\diag\bigl(\bB_{1}\bigr)\bigr)\trans d(\boldsymbol{\bar{\theta}}^{\ast})\bigr)\nn\\
			&+	\nu_{1k}\bigl(\bigl(\diag\bigl(\bB_{2}\herm\bigr)\trans d(\boldsymbol{\bar{\theta}})+\bigl(\diag\bigl(\bB_{2}\bigr)\bigr)\trans d(\boldsymbol{\bar{\theta}}^{\ast})\bigr)\label{sk0}
		\end{align}
		where
		\begin{align}
			\nu_{0k}&=2\tr(\bar{\boldsymbol{\Psi}}_{k})\tr\bigl(\bar{C}\bQ_{0k} \bR_{0k} 	\bR_{t}-\bar{C}\bQ_{1k}\bR_{0k}^{2}\bQ_{1k}	\bR_{t}\nn\\
			&+\bar{C}\bR_{0k}\bQ_{0k}	\bR_{t}\bigr),\\
			\nu_{1k}&=2\tr(\bar{\boldsymbol{\Psi}}_{k})\tr\bigl(\hat{\beta}_{1k}\bQ_{1k}\bR_{1k}\bR_{1}-\hat{\beta}_{1k}\bQ_{1k}\bR_{1k}^{2}\bQ_{1k}\bR_{1}\nn\\
			&+\hat{\beta}_{1k}\bR_{1k}\bQ_{1k}\bR_{1}\bigr).
		\end{align} 
		Note that in \eqref{sk0}, we have inserted \eqref{r0k} and \eqref{r1k}.
		
		From \eqref{sk0}, we obtain that 
		\begin{align}
			\nabla_{\bar{\thetv}}S_k&=\frac{\partial}{\partial{\bar{\thetv}^{\ast}}}S_{k}\nn\\
			&=\nu_{0k}\diag\bigl(\bB_{1}\bigr)+\nu_{1k}\diag\bigl(\bB_{2}\bigr),
		\end{align}
		which proves \eqref{derivbartheta}.
		
		In the case of  $\nabla_{\bar{\thetv}}I_k$, we focus on the differential of $ I_k $
		\begin{align}
			\!\!	d(I_{k})  &\!=\!\sum\nolimits _{i=1}^{K}\!\!\tr(d(\bar{\mathbf{R}}_{k})\bar{\boldsymbol{\Psi}}_{i})\!+\!\sum\nolimits _{i=1}^{K}\!\!\tr(\mathbf{R}_{k}d(\bar{\boldsymbol{\Psi}}_{i})\!)\!\nn\\
			&-\!2\tr\bigl(\bar{\boldsymbol{\Psi}}_{k}d(\bar{\boldsymbol{\Psi}}_{k})\!\bigr)\!+\!\frac{K\sigma^{2}}{\rho}\!\sum\nolimits _{i=1}^{K}\!\!\tr(d(\bar{\boldsymbol{\Psi}}_{i})\!)\\
			&\!=\tr(\tilde{\boldsymbol{\Psi}}d(\bar{\mathbf{R}}_{k}))-2\tr\bigl(\bar{\boldsymbol{\Psi}}_{k}d(\bar{\boldsymbol{\Psi}}_{k})\bigr)+\sum\nolimits _{i=1}^{K}\tr(\tilde{\mathbf{R}}_{k}d(\bar{\boldsymbol{\Psi}}_{i})),\label{ik1}
		\end{align}
		where  $\tilde{\boldsymbol{\Psi}}=\sum\nolimits _{i=1}^{K}\bar{\boldsymbol{\Psi}}_{i}$
		and $\tilde{\mathbf{R}}_{k}=\bar{\mathbf{R}}_{k}+\frac{K\sigma^{2}}{\rho}\mathbf{I}_{M}$. Use of \eqref{eq:dPsik1} into \eqref{ik1} yields
		\begin{align}
			&d(I_{k}) =\tr(\tilde{\boldsymbol{\Psi}}d(\bar{\mathbf{R}}_{k}))+\sum\nolimits _{i=1}^{K}\tr(\tilde{\mathbf{R}}_{k}\bigl(d(\bar{\mathbf{R}}_{i})\mathbf{Q}_{i}\bar{\mathbf{R}}_{i}\bigr))\nn\\
			&-\bar{\mathbf{R}}_{i}\mathbf{Q}_{i}d(\bar{\mathbf{R}}_{i})\mathbf{Q}_{i}\bar{\mathbf{R}}_{i}+\bar{\mathbf{R}}_{i}\mathbf{Q}_{i}d(\bar{\mathbf{R}}_{i})-2\tr\bigl(\bar{\boldsymbol{\Psi}}_{k}\bigl(d(\bar{\mathbf{R}}_{k})\mathbf{Q}_{k}\bar{\mathbf{R}}_{k}\nn\\
			&-\bar{\mathbf{R}}_{k}\mathbf{Q}_{k}d(\bar{\mathbf{R}}_{k})\mathbf{Q}_{k}\bar{\mathbf{R}}_{k}+\bar{\mathbf{R}}_{k}\mathbf{Q}_{k}d(\bar{\mathbf{R}}_{k})\bigr)\bigr)\nn\\
			& =\tr(\check{\boldsymbol{\Psi}}_{k}d(\bar{\mathbf{R}}_{k}))+\sum\nolimits _{i=1}^{K}\tr\bigl(\tilde{{\mathbf{R}}}_{ki}d(\bar{\mathbf{R}}_{i})\bigr)\\
			& =\tr(\check{\boldsymbol{\Psi}}_{k}(d(\bR_{0k})+d(\bR_{1k})))\nn\\
			&+\sum\nolimits _{i=1}^{K}\tr\bigl(\tilde{{\mathbf{R}}}_{ki}(d(\bR_{0i})+d(\bR_{1i}))\bigr),\label{dik1}
		\end{align}
		where
		\begin{align}
			\check{\bar{\boldsymbol{\Psi}}}_{k}&=\bar{\boldsymbol{\Psi}}-2\bigl(\mathbf{Q}_{k}\bar{\mathbf{R}}_{k}\bar{\boldsymbol{\Psi}}_{k}+\bar{\boldsymbol{\Psi}}_{k}\bar{\mathbf{R}}_{k}\mathbf{Q}_{k}-\mathbf{Q}_{k}\bar{\mathbf{R}}_{k}\bar{\boldsymbol{\Psi}}_{k}\bar{\mathbf{R}}_{k}\mathbf{Q}_{k}\bigr)\\
			\tilde{\mathbf{R}}_{ki}&=\mathbf{Q}_{i}\bar{\mathbf{R}}_{i}\tilde{\mathbf{R}}_{k}-\mathbf{Q}_{i}\bar{\mathbf{R}}_{i}\tilde{\mathbf{R}}_{k}\bar{\mathbf{R}}_{i}\mathbf{Q}_{i}+\bar{\tilde{\mathbf{R}}}_{k}\bar{\mathbf{R}}_{i}\mathbf{Q}_{i}.
		\end{align}
		Equation \eqref{dik1} is obtained because the  dependence of  $ \bar{\bR}_{k} $ from   $ \bPhi_{1}  $ is hidden only on $ \bR_{0k} $ and $ \bR_{1k} $. Substitution of  \eqref{r0k} and \eqref{r1k} into \eqref{dik} allows to prove $\nabla_{\thetv^{t}}I_k$ as
		\begin{align}
			\nabla_{\bar{\thetv}}I_k&=\frac{\partial}{\partial\boldsymbol{\bar{\theta}}^{\ast}}I_{k}  \nn\\
			&=\diag\bigl(\bar{\nu}_{1k}\bB_{1}+\bar{\nu}_{2k}\bB_{2}+\sum\nolimits_{i=1}^{K}(\tilde{\nu}_{ki1}\bB_{1}+\tilde{\nu}_{ki2}\bB_{2})\bigr),
		\end{align}
		where $\bar{\nu}_{1k}=\bar{C}	\tr\bigl(\check{\bar{\boldsymbol{\Psi}}}_{k}\bR_{t}\bigr)$,$\bar{\nu}_{2k}=\hat{\beta}_{1k}\tr\bigl(\check{\bar{\boldsymbol{\Psi}}}_{k}\bR_{1}\bigr)$, $\tilde{\nu}_{ki1}=\bar{C}\tr\bigl(\tilde{\mathbf{R}}_{ki}\bR_{t}\bigr)$, and $\tilde{\nu}_{ki2}=\hat{\beta}_{1k}\tr\bigl(\tilde{\mathbf{R}}_{ki}\bR_{1}\bigr)$.
		
		\section{Proof of Proposition~\ref{PropoGradientss}}\label{prop3}
		In the case of the STAR-RIS, we have to derive  $\nabla_{\boldsymbol{\theta^{t}}}\mathrm{SE}(\thetv,\betv) $ in terms of $ \boldsymbol{\theta}^{t\ast}$, which is written as
		\begin{equation}
			\nabla_{\thetv^{t}}\mathrm{SE}(\thetv,\betv)=c\sum_{k=1}^{K}\frac{I_{k}\nabla_{\boldsymbol{\theta}^{t}}S_{k}-S_{k}\nabla_{\boldsymbol{\theta}^{t}}I_{k}}{(1+\gamma_{k})I_{k}^{2}},
		\end{equation}
		where $c=\frac{\tau_{c}-\tau}{\tau_{c}\log_{2}(e)}$.
		
		Regarding the computation of $\nabla_{\thetv^{t}}S_k$, we observe that $\nabla_{\thetv^{t}}S_k = 0$ if $w_k=r$, i.e., when UE $k$ is in the reflection region.  Hence, we focus on the derivation of $\nabla_{\thetv^{t}}S_k$ when $w_k=t$, which requires the derivation of $ d(S_{k}) $. Similar to \eqref{eq:dSk}, we have
		\begin{align}
			d(S_{k})&=2\tr(\bar{\boldsymbol{\Psi}}_{k})\tr(d(\bar{\boldsymbol{\Psi}}_{k}))\label{eq:dSk1}
		\end{align}
		with $ 	\bar{\boldsymbol{\Psi}}_{k} $  written as in \eqref{psiK}. Note that, in $ \bar{\boldsymbol{\Psi}}_{k} $, only $ \bR_{0k} $ and $ \bR_{2k} $ depend on $ \bPhi_{2,w_{k}} $. Thus, its differential is obtained as
		\begin{align}
			d(\bar{\boldsymbol{\Psi}}_{k})&=d(\bR_{0k})\bQ_{0k}\bR_{0k}-\bR_{0k}\bQ_{1k}d(\mathbf{R}_{1k})\bQ_{1k}\bR_{0k}\nn\\
			&+\bR_{0k}\bQ_{0k}d(\bR_{0k})+d(\bR_{2k})\bQ_{2k}\bR_{2k}\nn\\
			&-\bR_{2k}\bQ_{2k}d(\mathbf{R}_{2k})\bQ_{2k}\bR_{2k}+\bR_{2k}\bQ_{2k}d(\bR_{2k}),\label{eq:dPsik2}
		\end{align} 
		where $ \bR_{0k}=\bar{\bD}\tr(\bB_{3t}\bPhi_{2,t}^{\H}) $ and  $ \bR_{2k}=\tr(\bB_{4t}\bPhi_{2,t}^{\H})\bR_{2} $  with $ \bB_{3t}=\bR_{1}\bPhi_{2,t}\bR_{2} $, $ \bar{\bD} =\tr(\bB_{1}\bPhi_{1}^{\H})\bR_{t} $, and $ \bB_{4t}=\hat{\beta}_{2k}\tr(\bR_{t} \bPhi_{2,t}\bR_{2}) $.
		
		The differentials of $ \bR_{0k} $ and $ \bR_{2k} $ are derived as
		\begin{align}
			d(\mathbf{R}_{0k})&=\bar{\bD}\bigl(\bigl(\diag\bigl(\bB_{3t}\herm\diag(\boldsymbol{{\beta}}^{t})\bigr)\bigr)\trans d(\boldsymbol{\theta}^{t})\nn\\
			&+\bigl(\diag\bigl(\bB_{3t}\diag(\boldsymbol{{\beta}}^{t})\bigr)\bigr)\trans d(\boldsymbol{\theta}^{t\ast})\bigr),\label{eq:dRk2}\\
			d(\mathbf{R}_{2k})&=\bR_{2}\bigl(\bigl(\diag\bigl(\bB_{4t}\herm\diag(\boldsymbol{{\beta}}^{t})\bigr)\bigr)\trans d(\boldsymbol{\theta}^{t})\nn\\
			&+\bigl(\diag\bigl(\bB_{4t}\diag(\boldsymbol{{\beta}}^{t})\bigr)\bigr)\trans d(\boldsymbol{\theta}^{t\ast})\bigr).\label{eq:dRk3}
		\end{align}
		
		Based on \eqref{eq:dRk2} and \eqref{eq:dRk3},  \eqref{eq:dSk1} becomes
		\begin{align}
			d(S_{k})&=2\tr(\bar{\boldsymbol{\Psi}}_{k})\tr\bigl(\bQ_{0k}\bR_{0k}d(\bR_{0k})-\bQ_{1k}\bR_{0k}^{2}\bQ_{1k}d(\mathbf{R}_{1k})\nn\\
			&+\bR_{0k}\bQ_{0k}d(\bR_{0k})+\bQ_{2k}\bR_{2k}d(\bR_{2k})\nn\\
			&-\bQ_{2k}\bR_{2k}^{2}\bQ_{2k}d(\mathbf{R}_{2k})+\bR_{2k}\bQ_{2k}d(\bR_{2k})\bigr)\\
			&=\nu_{2k}\bigl(\diag\bigl(\bigl(\bB_{3t}\herm\diag(\boldsymbol{\mathbf{\beta}}^{t})\bigr)\bigr)\trans d\boldsymbol{\theta}^t\nn\\
			&+\bigl(\diag\bigl(\bB_{3t}\diag(\boldsymbol{\mathbf{\beta}}^{t})\bigr)\bigr)\trans d\boldsymbol{\theta}^{t\ast}\bigr)\nn\\
			&+\nu_{3k}\bigl(\diag\bigl(\bigl(\bB_{4t}\herm\diag(\boldsymbol{\mathbf{\beta}}^{t})\bigr)\bigr)\trans d\boldsymbol{\theta}^t\nn\\
			&+\bigl(\diag\bigl(\bB_{4t}\diag(\boldsymbol{\mathbf{\beta}}^{t})\bigr)\bigr)\trans d\boldsymbol{\theta}^{t\ast}\bigr),\label{sk3}
		\end{align}
		where
		\begin{align}
			\nu_{2k}&\!=\!2\tr(\bar{\boldsymbol{\Psi}}_{k})\!\tr\bigl(\bQ_{0k} \bR_{0k} 	\bar{\bD}-\bQ_{1k}\bR_{0k}^{2}\bQ_{1k}	\bar{\bD}+\bR_{0k}\bQ_{0k}	\bar{\bD}\bigr),\\
			\nu_{3k}&\!=\!2\tr(\bar{\boldsymbol{\Psi}}_{k})\!\tr\bigl(\bQ_{2k}\bR_{2k}\bR_{2}\!-\!\bQ_{2k}\bR_{2k}^{2}\bQ_{2k}\bR_{2}\!+\!\bR_{2k}\bQ_{2k}\bR_{2}\bigr).
		\end{align} 
		
		Thus, in the case of $w_k=t$, we obtain
		\begin{align}
			\nabla_{\thetv^{t}}S_k&=\nu_{2k}\diag\bigl(\bB_{3t}\diag(\boldsymbol{{\beta}}^{t})\bigr)+\nu_{3k}\diag\bigl(\bB_{4t}\diag(\boldsymbol{{\beta}}^{t})\bigr),
		\end{align}
		which equals to \eqref{derivtheta_t}. The proof of \eqref{derivtheta_r} follows similar lines.
		
		To derive  $\nabla_{\bar{\thetv}}I_k$, we aim at finding the  differential of $ I_k $. In a similar way to  \eqref{ik1}, we obtain
		\begin{align}
			d(I_{k}) 	& =\tr(\check{\boldsymbol{\Psi}}_{k}(d(\bR_{0k})+d(\bR_{2k})))\nn\\
			&+\sum\nolimits _{i\in \mathcal{K}_t}\tr\bigl(\tilde{{\mathbf{R}}}_{ki}(d(\bR_{0i})+d(\bR_{2i}))\bigr),\label{dik}
		\end{align}
		
		Note that $d(\bR_{0i})=d(\bR_{2i})=0$ if $w_{i}\neq t$ because $ d(\bR_{0i}) $ and  $ d(\bR_{2i}) $  do not depend on $\thetv^{t}$ in this case. Hence, we have
		\begin{align}
			\nabla_{\thetv^{t}}I_k&=\frac{\partial}{\partial\boldsymbol{\thetv}^{t\ast}}I_{k} \nn\\
			&=\diag\bigl(\tilde{\mathbf{A}}_{kt}\diag(\boldsymbol{\mathbf{\beta}}^{t})\bigr),
		\end{align}
		where 
		\begin{equation}
			\tilde{\mathbf{A}}_{kt}=\begin{cases}
				\bar{\nu}_{2k}\bB_{3}+\bar{\nu}_{2k}\bB_{4}+\sum\nolimits _{i\in\mathcal{K}_{t}}^{K}(\tilde{\nu}_{ki2}\bB_{3}+\tilde{\nu}_{ki2}\bB_{4}) & w_{k}=t\\
				\sum\nolimits _{i\in\mathcal{K}_{t}}(\tilde{\nu}_{ki2}\bB_{3}+\tilde{\nu}_{ki2}\bB_{4}) & w_{k}\neq t
			\end{cases}
		\end{equation}
		with
		$\bar{\nu}_{1k}=	\tr\bigl(\check{\bar{\boldsymbol{\Psi}}}_{k}\bar{\bD}\bigr)$,$\bar{\nu}_{2k}=\hat{\beta}_{1k}\tr\bigl(\check{\bar{\boldsymbol{\Psi}}}_{k}\bR_{2}\bigr)$, $\tilde{\nu}_{ki1}=\tr\bigl(\tilde{\mathbf{R}}_{ki}\bar{\bD}\bigr)$, and $\tilde{\nu}_{ki2}=\hat{\beta}_{1k}\tr\bigl(\tilde{\mathbf{R}}_{ki}\bR_{2}\bigr)$. Thus, \eqref{derivtheta_t_Ik} is proved, while \eqref{derivtheta_r_Ik} is obtained by following the same steps.
		
		In the case of  $\nabla_{\boldsymbol{\beta}^{t}}S_{k}$, we consider first the scenario $w_{k}=t$, and we have
		\begin{align}
			d(\mathbf{R}_{0k})&=2\real\{ \bar{\bD}\bigl(\diag\bigl(\bB_{3t}\herm\diag(\boldsymbol{{\beta}}^{t})\bigr)\bigr)\}\trans d(\boldsymbol{\theta}^{t}),\label{eq:dRk4}\\
			d(\mathbf{R}_{2k})&=2\real\{\bR_{2}\bigl(\diag\bigl(\bB_{4t}\herm\diag(\boldsymbol{{\beta}}^{t})\bigr)\bigr)\}\trans d(\boldsymbol{\theta}^{t}).\label{eq:dRk5}
		\end{align}
		
		By substituting \eqref{eq:dRk4} and \eqref{eq:dRk5} into \eqref{eq:dSk1}, we obtain
		\begin{align}
			\nabla_{\thetv^{t}}S_k&=2\real\{ \bar{\bD}\bigl(\diag\bigl(\bB_{3t}\herm\diag(\boldsymbol{{\beta}}^{t})\bigr)\bigr)\}\nn\\
			&+2\real\{\bR_{2}\bigl(\diag\bigl(\bB_{4t}\herm\diag(\boldsymbol{{\beta}}^{t})\bigr)\bigr)\}.
		\end{align}
		
		Following the same lines as above, we result in
		\begin{align}
			\nabla_{\thetv^{r}}S_k&=2\real\{ \bar{\bD}\bigl(\diag\bigl(\bB_{3r}\herm\diag(\boldsymbol{{\beta}}^{r})\bigr)\bigr)\}\nn\\
			&+2\real\{\bR_{2}\bigl(\diag\bigl(\bB_{4r}\herm\diag(\boldsymbol{{\beta}}^{r})\bigr)\bigr)\},\\
			\nabla_{\thetv^{t}}I_k
			&=2\real\{\diag\bigl(\tilde{\mathbf{A}}_{kt}\diag(\boldsymbol{\mathbf{\beta}}^{t})\bigr)\},\\
			\nabla_{\thetv^{r}}I_k
			&=2\real\{\diag\bigl(\tilde{\mathbf{A}}_{kr}\diag(\boldsymbol{\mathbf{\beta}}^{r})\bigr)\},
		\end{align}
		which concludes the proof.
	\end{appendices}
	\bibliographystyle{IEEEtran}
	
	\bibliography{IEEEabrv,mybib}
	
\end{document}